\documentclass[10pt, a4paper,american]{article}
\usepackage[pdftex]{graphicx}
\usepackage{amsmath, amssymb, amsthm, amsfonts}
\usepackage{float, color}
\usepackage{ascmac, fancybox}
\usepackage{bm}

\usepackage{bbm}

\usepackage[margin=1in]{geometry}

\usepackage{mathrsfs}

\makeatletter

\theoremstyle{plain}
\numberwithin{equation}{section}
\theoremstyle{plain}
\newtheorem{thm}{Theorem}[section]
\theoremstyle{plain}

\theoremstyle{definition}

\theoremstyle{definition}
\newtheorem{example}[thm]{Example}
\theoremstyle{plain}
\newtheorem{lem}[thm]{Lemma}
\theoremstyle{plain}
\newtheorem{cor}[thm]{Corollary}
\theoremstyle{definition}

\makeatother

\usepackage{bm}% bold math
\usepackage[utf8]{inputenc}
\usepackage[T1]{fontenc}
\usepackage{mathptmx}

\usepackage{amsmath}
\usepackage{amsthm}
\usepackage{amssymb}
\usepackage{stmaryrd}

\usepackage{stmaryrd}
\usepackage{cases}
\usepackage{color}
\usepackage{comment}

\begin{document}
\title{Maximum logarithmic derivative bound on quantum state estimation as a dual of the Holevo bound}
\author{
	Koichi Yamagata%
	\thanks{koichi.yamagata@uec.ac.jp}\\
	{The University of Electro-Communications Department of Informatics,}\\
		{1-5-1, Chofugaoka, Chofu, Tokyo 182-8585, Japan} 
}%
%\date{\today}
\date{}

\maketitle

\begin{abstract}
In quantum estimation theory, the Holevo bound is known as a lower bound of weighed traces of covariances of unbiased estimators. 
The Holevo bound is defined by a solution of a minimization problem, and
in general, explicit solution is not known. When the dimension of
Hilbert space is two and the number of parameters is two, a explicit
form of the Holevo bound was given by Suzuki. In this paper, we focus
on a logarithmic derivative lies between the symmetric logarithmic
derivative (SLD) and the right logarithmic derivative (RLD) parameterized
by $\beta\in[0,1]$ to obtain lower bounds of weighted trace of covariance
of unbiased estimator. We introduce the maximum logarithmic derivative
bound as the maximum of bounds with respect to $\beta$. We show that
all monotone metrics induce lower bounds, and the maximum logarithmic
derivative bound is the largest bound among them. We show that the
maximum logarithmic derivative bound has explicit solution when the $d$ dimensional model has $d+1$ dimensional $\mathcal{D}$ invariant extension
of the SLD tangent space. Furthermore, when $d=2$, we show that the
maximization problem to define the maximum logarithmic derivative
bound is the Lagrangian duality of the minimization problem to define
Holevo bound, and is the same as the Holevo bound. This explicit solution
is a generalization of the solution for a two dimensional Hilbert
space given by Suzuki. We give also examples of families of quantum
states to which our theory can be applied not only for two dimensional
Hilbert spaces. 
\end{abstract}

\maketitle

\if 0
\theoremstyle{plain}
\newtheorem{thm}{Theorem}
\theoremstyle{definition}
\newtheorem{defn}[thm]{Definition}
\theoremstyle{plain}
\newtheorem{cor}[thm]{Corollary}
\fi

\if 0
\theoremstyle{defn}
%\newtheorem{defn}{Definition}[section]
\newtheorem{defn}{Definition}
\fi

\newcommand{\argmax}{\mathop{\rm arg~max}\limits}
\newcommand{\argmin}{\mathop{\rm arg~min}\limits}

\global\long\def\R{\mathbb{R}}%
\global\long\def\C{\mathbb{C}}%
\global\long\def\H{\mathcal{H}}%
\global\long\def\X{\mathcal{X}}%
\global\long\def\M{\mathcal{M}}%
\global\long\def\S{\mathcal{S}}%
\global\long\def\K{\mathcal{K}}%
\global\long\def\F{\mathcal{F}}%
\global\long\def\tr{{\rm tr\,}}%
\global\long\def\Tr{{\rm Tr\,}}%
\global\long\def\D{\mathcal{D}}%
\global\long\def\B{\mathcal{B}}%
\global\long\def\T{\mathcal{T}}%
\global\long\def\P{\mathcal{P}}%
\global\long\def\bra#1{\left\langle #1\right|}%
\global\long\def\ket#1{\left|#1\right\rangle }%
\global\long\def\braket#1#2{\left\langle #1\mid#2\right\rangle }%
\global\long\def\bR{\mathbf{R}}%
\global\long\def\bL{\mathbf{L}}%
\global\long\def\ii{\sqrt{-1}}%

\if 0
\begin{quotation}
The ``lead paragraph'' is encapsulated with the \LaTeX\ 
\verb+quotation+ environment and is formatted as a single paragraph before the first section heading. 
(The \verb+quotation+ environment reverts to its usual meaning after the first sectioning command.) 
Note that numbered references are allowed in the lead paragraph.
%
The lead paragraph will only be found in an article being prepared for the journal \textit{Chaos}.
\end{quotation}
\fi

\section{Introduction}

Let $\S=\left\{ \rho_{\theta};\,\theta\in\Theta\subset\R^{d}\right\} $
be a smooth parametric family of density operators on a Hilbert space
$\H$. An estimator is represented by a pair $(M,\hat{\theta})$ of
a POVM $M$ taking values on any finite set $\X$ and a map $\hat{\theta}:\X\to\Theta$.
An estimator $(M,\hat{\theta})$ is called unbiased if
\begin{equation}
E_{\theta}[M,\hat{\theta}]=\sum_{x\in\X}\hat{\theta}(x)\Tr\rho_{\theta}M_{x}=\theta\label{eq:unbias}
\end{equation}
is satisfied for all $\theta\in\Theta$. An estimator $(M,\hat{\theta})$
is called locally unbiased\cite{holevo} at a given point $\theta_{0}\in\Theta$
if the condition (\ref{eq:unbias}) is satisfied around $\theta_{0}$
up to the first order of the Taylor expansion, i.e.,
\begin{align}
\sum_{x\in\X}\hat{\theta}^{i}(x)\Tr\rho_{\theta_{0}}M_{x} & =\theta_{0}^{i}\qquad(i=1,\dots,d),\\
\sum_{x\in\X}\hat{\theta}^{i}(x)\Tr\partial_{j}\rho_{\theta_{0}}M_{x} & =\delta_{j}^{i}\qquad(i,j=1,\dots,d),
\end{align}
where $\partial_{j}\rho_{\theta_{0}}=\left.\frac{\partial}{\partial\theta^{j}}\rho_{\theta}\right|_{\theta=\theta_{0}}$.
It is well-known that the covariance matrix $V_{\theta_{0}}[M,\hat{\theta}]$
of an locally unbiased estimator $(M,\hat{\theta})$ at $\theta_{0}$
satisfies the following inequalities:

\begin{equation}
V_{\theta_{0}}[M,\hat{\theta}]\geq J_{\theta_{0}}^{(S)^{-1}},
\end{equation}
\begin{equation}
V_{\theta_{0}}[M,\hat{\theta}]\geq J_{\theta_{0}}^{(R)^{-1}},
\end{equation}
where $J_{\theta_{0}}^{(S)}:=\left[{\rm Re}\,(\Tr\rho_{\theta_{0}}L_{i}^{(S)}L_{j}^{(S)})\right]_{ij}$
is the symmetric logarithmic derivative (SLD) Fisher information matrix
at $\theta_{0}$ with SLDs $L_{i}^{(S)}$ ($1\leq i\leq d$) defined
by
\begin{equation}
\partial_{i}\rho_{\theta_{0}}=\frac{1}{2}\left(\rho_{\theta_{0}}L_{i}^{(S)}+L_{i}^{(S)}\rho_{\theta_{0}}\right),
\end{equation}
and $J_{\theta_{0}}^{(R)}:=\left[\Tr L_{i}^{(R)^{*}}\rho_{\theta_{0}}L_{j}^{(R)}\right]_{ij}$
is the right logarithmic derivative (RLD) Fisher information matrix
at $\theta_{0}$ with RLDs $L_{i}^{(R)}$ ($1\leq i\leq d$) defined
by
\begin{equation}
\partial_{i}\rho_{\theta_{0}}=\rho_{\theta_{0}}L_{i}^{(R)}.
\end{equation}
These matrix inequalities imply
\begin{equation}
\Tr GV_{\theta_{0}}[M,\hat{\theta}]\geq\Tr GJ_{\theta_{0}}^{(S)^{-1}}=:C_{\theta_{0},G}^{(S)},
\end{equation}
\begin{equation}
\Tr GV_{\theta_{0}}[M,\hat{\theta}]\geq\Tr GJ_{\theta_{0}}^{(R)^{-1}}+\Tr\left|\sqrt{G}{\rm Im}J_{\theta_{0}}^{(R)^{-1}}\sqrt{G}\right|=:C_{\theta_{0},G}^{(R)},
\end{equation}
for any $d\times d$ real positive matrix $G$, because
\begin{equation}
\min_{V}\{\Tr GV;\,V\geq J,V\text{ is a real matrix}\}=\Tr GJ+\Tr\left|\sqrt{G}{\rm Im}J\sqrt{G}\right|\label{eq:trabs}
\end{equation}
for any positive complex matrix $J$ (see Appendix \ref{sec:trabs_proof}
for the proof). 

A tighter lower bound of $\Tr GV_{\theta_{0}}[M,\hat{\theta}]$ than
the SLD bound $C_{\theta_{0},G}^{(S)}$ and the RLD bound $C_{\theta_{0},G}^{(R)}$
is known as the Holevo bound \cite{holevo} defined by
\begin{align}
C_{\theta_{0},G}^{(H)} & :=\min_{V,B}\left\{ \Tr GV;\,V\text{ is a real matrix such that }V\geq Z(B),\,Z_{ij}(B)=\Tr\rho_{\theta_{0}}B_{j}B_{i},\right.\\
 & \qquad B_{1},\dots,B_{d}\text{ are Hermitian operators on \ensuremath{\H\ }such that }\Tr\partial_{i}\rho_{\theta_{0}}B_{j}=\delta_{ij}\},\label{eq:holevo_bound}
\end{align}
(see Appendix \ref{sec:Holevo_proof} for the derivation and the proof)
and it satisfies
\begin{equation}
\Tr GV_{\theta_{0}}[M,\hat{\theta}]\geq C_{\theta_{0},G}^{(H)}\geq\max(C_{\theta_{0},G}^{(S)},C_{\theta_{0},G}^{(R)}).
\end{equation}
It is known that the Holevo bound is asymptotically achievable in
theory of quantum local asymptotic normality \cite{YFG,qlan2,guta}.
Note that the minimization problem over $V$ in (\ref{eq:holevo_bound})
is explicitly solved by using (\ref{eq:trabs}), and
\begin{align}
C_{\theta_{0},G}^{(H)} & =\min_{B}\left\{ \Tr GZ(B)+\Tr\left|\sqrt{G}{\rm Im}Z(B)\sqrt{G}\right|;\,Z_{ij}(B)=\Tr\rho_{\theta_{0}}B_{j}B_{i},\right.\label{eq:holevo_bound2}\\
 & \qquad B_{1},\dots,B_{d}\text{ are Hermitian operators on \ensuremath{\H\ }such that }\Tr\partial_{i}\rho_{\theta_{0}}B_{j}=\delta_{ij}\}.\nonumber 
\end{align}
However, the minimization problem over $B$ in (\ref{eq:holevo_bound2})
is not trivial in general. Suzuki\cite{suzukiHolevo} showed that,
when $\dim\H=2$ and $d=2$, the Holevo bound can be represented explicitly
by using the SLD bound and the RLD bound as
\begin{equation}
C_{\theta_{0},G}^{(H)}=\begin{cases}
C_{\theta_{0},G}^{(R)} & \text{if }C_{\theta_{0},G}^{(R)}\geq\frac{C_{\theta_{0},G}^{(Z)}+C_{\theta_{0},G}^{(S)}}{2},\\
C_{\theta_{0},G}^{(R)}+S_{\theta_{0},G} & \text{otherwise, }
\end{cases}\label{eq:suzuki_bound_ori}
\end{equation}
where $C_{\theta_{0},G}^{(Z)}$ and $S_{\theta_{0},G}$ are positive
values defined by
\begin{equation}
C_{\theta_{0},G}^{(Z)}:=\Tr GZ(L^{(S)})+\Tr\left|\sqrt{G}{\rm Im}Z(L^{(S)})\sqrt{G}\right|,
\end{equation}
\begin{equation}
S_{\theta_{0},G}:=\frac{\left[\frac{1}{2}(C_{\theta_{0},G}^{(Z)}+C_{\theta_{0},G}^{(S)})-C_{\theta_{0},G}^{(R)}\right]^{2}}{C_{\theta_{0},G}^{(Z)}-C_{\theta_{0},G}^{(R)}}.
\end{equation}

In this paper, we focus on a logarithmic derivative $L_{i}^{(\beta)}$
lies between SLD $L_{i}^{(S)}$ and RLD $L_{i}^{(R)}$, that defined
by
\begin{equation}
\partial_{i}\rho_{\theta_{0}}=\frac{(1+\beta)}{2}\rho_{\theta_{0}}L_{i}^{(\beta)}+\frac{(1-\beta)}{2}L_{i}^{(\beta)}\rho_{\theta_{0}}
\end{equation}
with $\beta\in[0,1]$. When $\beta=0$, $L_{i}^{(\beta)}$ coincides
with SLD $L_{i}^{(S)}$, and when $\beta=1$, $L_{i}^{(\beta)}$ coincides
with RLD $L_{i}^{(R)}$. The Fisher information matrix with respect
to $\left\{ L_{i}^{(\beta)}\right\} _{i=1}^{d}$ is
\begin{equation}
J_{\theta_{0}}^{(\beta)}:=\left[\Tr\partial_{i}\rho_{\theta_{0}}L_{j}^{(\beta)}\right]_{ij},
\end{equation}
and we show that the inequalities
\begin{equation}
V_{\theta_{0}}[M,\hat{\theta}]\geq J_{\theta_{0}}^{(\beta)^{-1}}
\end{equation}
and
\begin{equation}
\Tr GV_{\theta_{0}}[M,\hat{\theta}]\geq\Tr GJ_{\theta_{0}}^{(\beta)^{-1}}+\Tr\left|\sqrt{G}{\rm Im}J_{\theta_{0}}^{(\beta)^{-1}}\sqrt{G}\right|=:C_{\theta_{0},G}^{(\beta)},
\end{equation}
for the covariance matrix $V_{\theta_{0}}[M,\hat{\theta}]$ of any
locally unbiased estimator $(M,\hat{\theta})$ at $\theta_{0}$. We
call
\begin{equation}
\max_{0\leq\beta\leq1}C_{\theta_{0},G}^{(\beta)}\label{eq:max_prob}
\end{equation}
a maximum logarithmic derivative bound. 
More generally, monotone metrics
introduced by Petz can also induce Fisher information matrices and
lower bounds of $\Tr GV_{\theta_{0}}[M,\hat{\theta}]$\cite{monotone,cptni}.
We show that the maximum logarithmic derivative bound is the largest
bound among them.

The maximization problem (\ref{eq:max_prob}) is also not trivial
in general. However, when the model $\S=\left\{ \rho_{\theta};\,\theta\in\Theta\subset\R^{d}\right\} $
has $d+1$ dimensional real space $\tilde{\T}\supset{\rm span}_{\R}\{L_{i}^{(S)}\}_{i=1}^{d}$
such that $\D_{\rho_{\theta_{0}}}(\tilde{\T})\subset\tilde{\T}$ at
$\theta_{0}\in\Theta$, we show that $\max_{0\leq\beta\leq1}C_{\theta_{0},G}^{(\beta)}$
has explicit solution:
\begin{equation}
\max_{0\leq\beta\leq1}C_{\theta_{0},G}^{(\beta)}=\begin{cases}
C_{\theta_{0},G}^{(1)} & \text{if }\hat{\beta}\geq1,\\
C_{\theta_{0},G}^{(\hat{\beta})} & \text{\text{otherwise}},
\end{cases}\label{eq:mld_expli}
\end{equation}
with
\begin{equation}
\hat{\beta}=\begin{cases}
\frac{\Tr\left|\sqrt{G}{\rm Im}J_{\theta_{0}}^{(R)^{-1}}\sqrt{G}\right|}{2\Tr G\left\{ J_{\theta_{0}}^{(S)^{-1}}-{\rm Re}(J_{\theta_{0}}^{(R)^{-1}})\right\} } & \text{if }J_{\theta_{0}}^{(S)^{-1}}\not={\rm Re}(J_{\theta_{0}}^{(R)^{-1}}),\\
\infty & \text{\text{otherwise}},
\end{cases}
\end{equation}
where $\D_{\rho_{\theta_{0}}}$ is the commutation operator (see Section
\ref{sec:rewrite_holevo}). Furthermore, when $d=2$, we show that
the maximization problem (\ref{eq:max_prob}) is the Lagrangian duality
of the minimization problem to define Holevo bound, thus
\begin{equation}
\max_{0\leq\beta\leq1}C_{\theta_{0},G}^{(\beta)}=C_{\theta_{0},G}^{(H)}.
\end{equation}
Actually, the explicit solution (\ref{eq:mld_expli}) is a generalization
of the solution (\ref{eq:suzuki_bound_ori}) for $\dim\H=2$ 
{\color{black}(see Appendix \ref{sec:suzuki_mld})}.
%given in \cite{suzukiHolevo}.

This paper is organized as follows. In Section \ref{sec:mld}, we
introduce logarithmic derivatives and Fisher information matrices
induced by monotone metrics, and we derive the maximum logarithmic
derivative bound. In Section \ref{sec:rewrite_holevo}, we introduce
a commutation operator $\D$, and the Holevo bound is rewritten in
simpler form by using a $\D$ invariant space of Hermitian operators.
In Section \ref{sec:codimension}, we show that $\max_{0\leq\beta\leq1}C_{\theta_{0},G}^{(\beta)}$
has explicit solution (\ref{eq:mld_expli}) when the $d$ dimensional
model has $d+1$ dimensional $\D$ invariant extension of SLD tangent
space. Further, we show the maximum logarithmic derivative bound is
the same as the Holevo bound if $d=2$. At the end of the section,
we give examples of families of quantum states to which our theory
can be applied not only for $\dim\H=2$. Section \ref{sec:Conclusion}
is the conclusion. For the reader’s convenience, some additional material
is presented in the Appendix. In Appendix \ref{sec:trabs_proof},
a proof of (\ref{eq:trabs}) is given. In Appendix \ref{sec:Holevo_proof},
a brief proof and derivation of Holevo bound is presented. In Appendix
\ref{sec:schur}, the Schur complement, which plays an important role,
is introduced. 
In Appendix \ref{sec:Dinv}, details of the commutation
operator $\D$ as a tool for multiple inner products is described.
{\color{black} In Appendix \ref{sec:suzuki_mld}, 
it is shown that the explicit form (\ref{eq:suzuki_bound_ori})
%given 
for $\dim\H=2$ and $d=2$ can be derived from (\ref{eq:mld_expli}). 
}

\section{Maximum logarithmic derivative bound\label{sec:mld}}

Let $\S=\left\{ \rho_{\theta};\,\theta\in\Theta\subset\R^{d}\right\} $
be a smooth parametric family of density operators on a finite dimensional
Hilbert space $\H$. The covariance matrix $V_{\theta_{0}}[M,\hat{\theta}]$
of an locally unbiased estimator $(M,\hat{\theta})$ at $\theta_{0}$
satisfies the classical Cram\'{e}r Rao inequality,
\begin{equation}
V_{\theta_{0}}[M,\hat{\theta}]\geq J_{\theta_{0}}^{(M)^{-1}}
\end{equation}
where
\begin{equation}
J_{\theta_{0}}^{(M)}=\left[\sum_{x\in\X}\frac{\left(\Tr\partial_{i}\rho_{\theta_{0}}M_{x}\right)\left(\Tr\partial_{j}\rho_{\theta_{0}}M_{x}\right)}{\Tr\rho_{\theta_{0}}M_{x}}\right]_{ij}
\end{equation}
is the classical Fisher information matrix with respect to the POVM
$M$. The equality is achieved when
\begin{equation}
\hat{\theta}^{i}(x)=\theta_{0}^{i}+\sum_{j=1}^{d}\left[J_{\theta_{0}}^{(M)^{-1}}\right]^{ij}\frac{\Tr\partial_{j}\rho_{\theta_{0}}M_{x}}{\Tr\rho_{\theta_{0}}M_{x}}.
\end{equation}
Thus the minimization of $\Tr GV_{\theta_{0}}[M,\hat{\theta}]$ is
reduced to the minimization of $\Tr GJ_{\theta_{0}}^{(M)^{-1}}$ for
any $d\times d$ real positive matrix $G$. In this section, we consider
lower bounds of $\Tr GJ_{\theta_{0}}^{(M)^{-1}}$ directly induced
by monotone metrics. 

Let $P:(0,\infty)\rightarrow(0,\infty)$ be an operator monotone function
such that $P(1)=1$ \cite{Bhatia}. Let $\B(\H)$ be the set of all
linear operators on $\H$. A monotone metric at $\theta_{0}\in\Theta$
is an inner product $K_{\rho_{\theta_{0}}}^{(P)}(\cdot,\cdot)$ on
$\B(\H)$ defined by
\begin{equation}
K_{\rho_{\theta_{0}}}^{(P)}(X,Y)=\Tr X^{*}\left[(\bR_{\rho_{\theta_{0}}}P(\bL_{\rho_{\theta_{0}}}\bR_{\rho_{\theta_{0}}}^{-1}))^{-1}Y\right],
\end{equation}
where $\bL_{\rho_{\theta_{0}}}$ and $\bR_{\rho_{\theta_{0}}}$ are
super operators on $\B(\H)$ defined by
\begin{align}
\bL_{\rho_{\theta_{0}}}(X) & =\rho_{\theta_{0}}X,\\
\bR_{\rho_{\theta_{0}}}(X) & =X\rho_{\theta_{0}},
\end{align}
with a strictly positive operator $\rho_{\theta_{0}}\in\S$ \cite{monotone,cptni}.
Note that $\bL_{\rho_{\theta_{0}}}$ and $\bR_{\rho_{\theta_{0}}}$
are commutative so the super operator $(\bR_{\rho_{\theta_{0}}}P(\bL_{\rho_{\theta_{0}}}\bR_{\rho_{\theta_{0}}}^{-1}))^{-1}$
is well-defined. The monotone metric $K_{\rho_{\theta_{0}}}^{(P)}(\cdot,\cdot)$
has the monotonicity
\begin{equation}
K_{\rho_{\theta_{0}}}^{(P)}(X,X)\geq K_{T(\rho_{\theta_{0}})}^{(P)}(T(X),T(X))\label{eq:monotonicity}
\end{equation}
under any channel $T:\B(\H)\to\B(\H')$ mapping to another Hilbert
space $\H'$. 

The logarithmic derivative $L_{i}^{(P)}$ and the Fisher information
matrix $J_{\theta_{0}}^{(P)}$ with respect to $P$ are
\begin{align}
L_{i}^{(P)} & =(\bR_{\theta_{0}}P(\bL_{\theta_{0}}\bR_{\theta_{0}}^{-1}))^{-1}\partial_{i}\rho_{\theta_{0}}\qquad(i=1,\dots,d),\\
J_{\theta_{0}}^{(P)} & =\left[K_{\rho_{\theta_{0}}}^{(P)}(\partial_{i}\rho_{\theta_{0}},\partial_{j}\rho_{\theta_{0}})\right]_{ij}=\left[\Tr\partial_{i}\rho_{\theta_{0}}L_{j}^{(P)}\right]_{ij}.
\end{align}
Because a {\color{black} linear} map $T^{(M)}:X\mapsto{\rm Diag}(\Tr XM_{1},\Tr XM_{2},\dots,\Tr XM_{|\X|})$
($X\in\B(\H)$) is a quantum channel, for any POVM $M$ taking values
on $\X=\{1,2,\dots,|\X|\},$ the monotonicity (\ref{eq:monotonicity})
of $K_{\rho_{\theta_{0}}}^{(P)}(\cdot,\cdot)$ induces a matrix inequality
\begin{equation}
J_{\theta_{0}}^{(P)}\geq\left[K_{T^{(M)}(\rho_{\theta_{0}})}^{(P)}(T^{(M)}(\partial_{i}\rho_{\theta_{0}}),T^{(M)}(\partial_{j}\rho_{\theta_{0}}))\right]_{ij}=J_{\theta_{0}}^{(M)},\label{eq:monotonicity_POVM}
\end{equation}
where ${\rm Diag}(\cdots)$ indicates a diagonal matrix.  
The inequality
(\ref{eq:monotonicity_POVM}) implies
\begin{align}
\Tr GJ_{\theta_{0}}^{(M)^{-1}} & \geq\min\{\Tr GV;\,V\geq J_{\theta_{0}}^{(P)^{-1}},V\text{ is a real matrix}\}\\
 & =\Tr GJ_{\theta_{0}}^{(P)^{-1}}+\Tr\left|\sqrt{G}{\rm Im}J_{\theta_{0}}^{(P)^{-1}}\sqrt{G}\right|=:C_{\theta_{0},G}^{(P)},
\end{align}
due to (\ref{eq:trabs}). 
{\color{black}
To obtain a tighter bound, we consider 
maximizing $C_{\theta_{0},G}^{(P)}$ with respect to $P$.
In existing studies, such maximization was considered
in several models, 
and SLD or RLD bounds were derived\cite{block_pure}.
In this section, we consider maximization of
$C_{\theta_{0},G}^{(P)}$ in general models.
}

In quantum state estimation, 
a family of linear functions
\begin{equation}
\F=
\left\{ P^{(\beta)}(x)=\frac{1+\beta}{2}x+\frac{1-\beta}{2};\,\beta\in[-1,1]\right\} \label{eq:beta_family},
\end{equation}
is particularly important for monotone metrics among operator monotone
functions
{\color{black}
because 
the operator monotone function $P$ which
maximize the lower bound $C^{(P)}_{\theta_0}$ is always 
in the family $\F$ of functions 
%\eqref{eq:beta_family}
as we will show in Theorem \ref{thm:monotone_bound}.
%the lower bound $C^{(P^{(\beta)})}_{\theta_0}$ corresponding to $P^{(\beta)}$ is tighter than  
%\begin{equation}
%$J^{(P^{(\beta)})}_{\theta_0} 
%\leq
%J^{(P)}_{\theta_0} $
%for any operator monotone function $P$ such 
%\end{equation}
}
We write $K_{\rho_{\theta_{0}}}^{(\beta)}$, $L_{i}^{(\beta)}$,
$J_{\theta_{0}}^{(\beta)}$, and $C_{\theta_{0},G}^{(\beta)}$ instead
of $K_{\rho_{\theta_{0}}}^{(P^{(\beta)})}$, $L_{i}^{(P^{(\beta)})}$,
$J_{\theta_{0}}^{(P^{(\beta)})}$, and $C_{\theta_{0},G}^{(P^{(\beta)})}$.
The $\beta$ logarithmic derivative $L_{i}^{(\beta)}$ is defined
by
\begin{equation}
\partial_{i}\rho_{\theta_{0}}=\frac{(1+\beta)}{2}\rho_{\theta_{0}}L_{i}^{(\beta)}+\frac{(1-\beta)}{2}L_{i}^{(\beta)}\rho_{\theta_{0}}.\label{eq:Lbeta_def}
\end{equation}
Let $\left\langle \cdot,\cdot\right\rangle ^{(\beta)}$ be an inner
product on $\B(\H)$ defined by
\begin{equation}
\left\langle X,Y\right\rangle ^{(\beta)}=\frac{1}{2}\Tr X^{*}\left\{ (1+\beta)\rho_{\theta_{0}}Y+(1-\beta)Y\rho_{\theta_{0}}\right\} .\label{eq:beta_inner}
\end{equation}
By using this inner product, the $\beta$ Fisher information matrix
can be written as
\begin{equation}
J_{\theta_{0}}^{(\beta)}=\left[K_{\rho_{\theta_{0}}}^{(\beta)}(\partial_{i}\rho_{\theta_{0}},\partial_{j}\rho_{\theta_{0}})\right]_{ij}=\left[\left\langle L_{i}^{(\beta)},L_{j}^{(\beta)}\right\rangle ^{(\beta)}\right]_{ij}.\label{eq:beta_fisher}
\end{equation}
Note that $L_{i}^{(0)}$ coincides with SLD $L_{i}^{(S)}$, and $L_{i}^{(1)}$
coincides with RLD $L_{i}^{(R)}$. 

{\color{black} Let us prove the optimality of the family $\F$ of functions.
Because any operator monotone function
$P:(0,\infty)\rightarrow(0,\infty)$ is differentiable and concave\cite{Bhatia},
there always exists $\beta\in[-1,1]$ such that 
$\left.\frac{{\rm d}}{{\rm d}x}P(x)\right|_{x=1}=\frac{1+\beta}{2}$
if $P(1)=1$.
Since the line $P^{(\beta)}(x)$ is a tangent to the concave function $P(x)$ at $x=1$, 
%if $\left.\frac{{\rm d}}{{\rm d}x}P(x)\right|_{x=1}=\frac{1+\beta}{2}$
%and $P(1)=1$ with $\beta\in[-1,1]$, then
\begin{equation}
P(x)\leq P^{(\beta)}(x)
\end{equation}
for any $x\in(0,\infty)$, 
and this implies
\begin{equation}
J_{\theta_{0}}^{(P)}\geq J_{\theta_{0}}^{(\beta)}
\end{equation}
and
\begin{equation}
C_{\theta_{0},G}^{(\beta)}\geq C_{\theta_{0},G}^{(P)}.
\end{equation}
Further, $C_{\theta_{0},G}^{(\beta)}=C_{\theta_{0},G}^{(-\beta)}$
because $L_{i}^{(\beta)}$ is conjugate transpose of $L_{i}^{(-\beta)}$.
Therefore, we do not need to consider operator monotone functions
other than $P^{(\beta)}$ for $\beta\in[0,1]$. 
}

Collecting these results, we have the following theorem. 
\begin{thm}\label{thm:monotone_bound}
For any locally unbiased estimator $(M,\hat{\theta})$ at $\theta_{0}$,
a $d\times d$ real positive matrix $G$, and an operator monotone
function $P:(0,\infty)\to(0,\infty)$ such that $P(1)=1$ and $\left.\frac{{\rm d}}{{\rm d}x}P(x)\right|_{x=1}=\frac{1+\beta}{2}$
with $\beta\in[-1,1]$,
\begin{equation}
\Tr GV_{\theta_{0}}[M,\hat{\theta}]\geq\Tr GJ_{\theta_{0}}^{(M)^{-1}}\geq C_{\theta_{0},G}^{(\left|\beta\right|)}\geq C_{\theta_{0},G}^{(P)}.
\end{equation}
\end{thm}

From this theorem, we have an inequality
\begin{equation}
\Tr GV_{\theta_{0}}[M,\hat{\theta}]\geq\max_{0\leq\beta\leq1}C_{\theta_{0},G}^{(\beta)},
\end{equation}
and we call the RHS of this inequality the maximum logarithmic derivative
bound. 

\section{Equivalent expressions of Holevo bound\label{sec:rewrite_holevo}}

In this section, we give a simpler form of the Holevo bound by using
a commutation operator. Let $\S=\left\{ \rho_{\theta};\,\theta\in\Theta\subset\R^{d}\right\} $
be a smooth parametric family of density operators on a finite dimensional
Hilbert space $\H$. Let $\D_{\rho_{\theta_{0}}}:\B(\H)\to\B(\H)$
be the commutation operator with respect to a faithful state $\rho_{\theta_{0}}\in\S$
on the set of linear operators $\B(\H)$ on $\H$ defined by
\begin{equation}
\D_{\rho_{\theta_{0}}}(X)\rho_{\theta_{0}}+\rho_{\theta_{0}}\D_{\rho_{\theta_{0}}}(X)=\sqrt{-1}(X\rho_{\theta_{0}}-\rho_{\theta_{0}}X),\label{eq:D_def}
\end{equation}
for $X\in\B(\H)$ \cite{holevo}. The commutation operator can also
be defined by
\begin{equation}
\D_{\rho_{\theta_{0}}}=\frac{1}{\ii}(\bL_{\rho_{\theta_{0}}}-\bR_{\rho_{\theta_{0}}})(\bL_{\rho_{\theta_{0}}}+\bR_{\rho_{\theta_{0}}})^{-1}.
\end{equation}
 When $X$ is a Hermitian operator, $\D_{\rho_{\theta_{0}}}(X)$
is also a Hermitian operator. Through the commutation operator, the
$\beta$ logarithmic derivatives $\{L_{i}^{(\beta)}\}_{i=1}^{d}$
and the corresponding inner product are linked by the following relations:
\begin{align}
L_{i}^{(\beta)} & =(I+\beta\sqrt{-1}\D_{\rho_{\theta_{0}}})^{-1}(L_{i}^{(0)})\qquad(i=1,\dots,d),\label{eq:bLD_SLD}\\
\left\langle A,B\right\rangle^{(\beta)} & =\left\langle A,(I+\beta\sqrt{-1}\D_{\rho_{\theta_{0}}})(B)\right\rangle ^{(0)}\qquad(A,B\in\B(\H))\label{eq:bLD_SLD_inner}
\end{align}
for $\beta\in[0,1]$. Note that $I+\beta\ii\D_{\rho_{\theta_{0}}}$
is invertible for $\rho_{\theta_{0}}>0$ since the operator norm of
$\D_{\rho_{\theta_{0}}}$ is \\
$\max\{\left|\frac{\lambda-\mu}{\lambda+\mu}\right|;\lambda,\mu\text{ are eigenvalues of }\text{\ensuremath{\rho_{\theta_{0}}}}\}<1$.
For details about the commutation operator $\D_{\rho_{\theta_{0}}}$,
see Appendix \ref{sec:Dinv}. By considering a $\D_{\rho_{\theta_{0}}}$
invariant extension $\tilde{\T}\supset\T$ of the SLD tangent space
$\T:={\rm span}_{\R}\,\{L_{i}^{(S)}\}_{i=1}^{d}$, the minimization
problem to define the Holevo bound is simplified as follows.
\begin{thm}
\label{thm:rewrite_holevo}Suppose that a quantum statistical model
$\S=\left\{ \rho_{\theta};\theta\in\Theta\subset\R^{d}\right\} $
on $\H$ has a $\D_{\rho_{\theta_{0}}}$ invariant extension $\tilde{\T}$
of the SLD tangent space of $\S$ at $\theta=\theta_{0}$. Let $\{D_{j}^{(S)}\}_{j=1}^{r}$
be a basis of $\tilde{\T}$. The Holevo bound defined by (\ref{eq:holevo_bound})
is rewritten as
\begin{align}
C_{\theta_{0},G}^{(H)} & =\min_{F}\left\{ \Tr GZ+\Tr\left|\sqrt{G}{\rm Im}Z\sqrt{G}\right|;\,\,Z= F^{\intercal}\Sigma F,\right.\\
 & \qquad F\text{ is an \ensuremath{r\times d} real matrix satisfying } F^{\intercal}{\rm \,Re}(\tau)=I\},\label{eq:min_F}
\end{align}
where $\Sigma$ and $\tau$ are $r\times r$ and $r\times d$ complex
matrix whose $(i,j)th$ entries are given by $\Sigma_{ij}=\Tr\rho_{\theta_{0}}D_{j}^{(S)}D_{i}^{(S)}$
and $\tau_{ij}=\Tr\rho_{\theta_{0}}L_{j}^{(S)}D_{i}^{(S)}$. 
\end{thm}

\begin{proof}
Let $\tilde{\T}^{\perp}$ be the orthogonal complement of $\tilde{\T}$
in the set $\B_{h}(\H)$ of Hermitian operators with respect to the
inner product $\left\langle \cdot,\cdot\right\rangle ^{(0)}$,  and
let $\P:\B_{h}(\H)\to\tilde{\T}$ and $\P^{\perp}:\B_{h}(\H)\to\tilde{\T}^{\perp}$
be the projections associated with the decomposition $\B_{h}(\H)=\tilde{\T}\oplus\tilde{\T}^{\perp}$.
For $X\in\tilde{\T}^{\perp}$ and $Y\in\tilde{\T}$, 
\begin{align}
\Tr X\rho_{\theta_{0}}Y & =\left\langle X,Y\right\rangle ^{(1)}=\left\langle X,(I+\sqrt{-1}\D_{\rho_{\theta_{0}}})(Y)\right\rangle ^{(0)}=0.\label{eq:SLD_RLD_orth}
\end{align}

Let $\{B_{j}\}_{j=1}^{d}$ be observables achieving the minimum in
(\ref{eq:holevo_bound}). $\{\P(B_{j})\}_{j=1}^{d}$ also satisfies
the local unbiasedness condition
\begin{equation}
\Tr\partial_{i}\rho_{\theta_{0}}\P(B_{j})=\left\langle L_{i}^{(S)},\P(B_{j})\right\rangle ^{(0)}=\left\langle L_{i}^{(S)},B_{j}\right\rangle ^{(0)}=\delta_{ij}.
\end{equation}
Further, because of (\ref{eq:SLD_RLD_orth}),
\begin{align}
Z_{ij}(B) & =\Tr B_{i}\rho_{\theta_{0}}B_{j}=\Tr\left\{ \P(B_{i})+\P^{\perp}(B_{i})\right\} \rho_{\theta_{0}}\left\{ \P(B_{j})+\P^{\perp}(B_{j})\right\} \\
 & =\Tr\P(B_{i})\rho_{\theta_{0}}\P(B_{j})+\Tr\P^{\perp}(B_{i})\rho_{\theta_{0}}\P^{\perp}(B_{j})=Z_{ij}(\P(B))+Z_{ij}(\P^{\perp}(B)).
\end{align}
This decomposition implies $Z(B)\geq Z(\P(B))$, thus $\{B_{j}\}_{j=1}^{d}\subset\tilde{\T}$. 

Observables $\{B_{j}\}_{j=1}^{d}\subset\tilde{\T}$ can be expressed
by $B_{j}=\sum_{k}F_{j}^{k}D_{k}^{(S)}$ with an $r\times d$ real
matrix $F$. By using $F$, the local unbiasedness condition in (\ref{eq:holevo_bound})
is written as
\begin{equation}
\left\langle L_{i}^{(S)},B_{j}\right\rangle ^{(0)}=F_{j}^{k}\left\langle L_{i}^{(S)},D_{k}^{(S)}\right\rangle ^{(0)}=F_{j}^{k}\left({\rm Re}\tau\right)_{ik}=\delta_{ij},
\end{equation}
and $Z(B)$ is written as $Z(B)=F^{\intercal}\Sigma F$. 
\end{proof}
Due to this theorem, we can easily see $C_{\theta_{0},G}^{(H)}\{\rho_{\theta}^{\otimes n}\}=\frac{1}{n}C_{\theta_{0},G}^{(H)}\{\rho_{\theta}\}$.
In this paper, we use further rewrite of the Holevo bound as follows. 
\begin{cor}
\label{cor:easy_holevo2}Suppose $D_{i}^{(S)}=L_{i}^{(S)}$ for $1\leq i\leq d$
in Theorem \ref{thm:rewrite_holevo}, and let $R=\left({\rm Re}\Sigma\right)^{-1}\Sigma\left({\rm Re}\Sigma\right)^{-1}=\begin{pmatrix}R_{1} & R_{2}^{*}\\
R_{2} & R_{3}
\end{pmatrix}$ with $d\times d$, $(r-d)\times d$, and $(r-d)\times(d-d)$ block
matrices $R_{1},R_{2},$ and $R_{3}$.  The Holevo bound is rewritten
as
\begin{align}
C_{\theta_{0},G}^{(H)} & =\min_{f}\left\{ \Tr GZ(f)+\Tr\left|\sqrt{G}{\rm Im}Z(f)\sqrt{G}\right|;\right.\\
 & \qquad f\text{ is an \ensuremath{(r-d)\times d} real matrix}\},
\end{align}
where
\begin{equation}
Z(f)=(I,f^{\intercal})R\begin{pmatrix}I\\
f
\end{pmatrix}=R_{1}+R_{2}^{*}f+f^{*}R_{2}+f^{*}R_{3}f.
\end{equation}
\end{cor}

\begin{proof}
Let $\Sigma=\begin{pmatrix}\Sigma_{1} & \Sigma_{2}^{*}\\
\Sigma_{2} & \Sigma_{3}
\end{pmatrix}$ be partitioned in the same manner as $R$. For an $r\times d$ real
matrix $F$, the condition $F^{\intercal}{\rm \,Re}(\tau)=I$ implies
\begin{equation}
F{\rm Re}\Sigma=\begin{pmatrix}I\\
f
\end{pmatrix}
\end{equation}
with an $(r-d)\times d$ real matrix $f$ because $\tau=\begin{pmatrix}\Sigma_{1}\\
\Sigma_{2}
\end{pmatrix}$. By using $f$, $F^{\intercal}\Sigma F$ in Theorem \ref{thm:rewrite_holevo}
can be written as
\begin{align}
F^{\intercal}\Sigma F & =(I,f^{\intercal})\left({\rm Re}\Sigma\right)^{-1}\Sigma\left({\rm Re}\Sigma\right)^{-1}\begin{pmatrix}
I\\
f
\end{pmatrix}=(I,f^{\intercal})R\begin{pmatrix}
I\\
f
\end{pmatrix}
\end{align}
\end{proof}
Actually, $R$ in Corollary \ref{cor:easy_holevo2} coincides with
the inverse RLD Fisher information matrix of a supermodel $\tilde{\S}\supset\S$
of $\S$ that have SLDs $\{D_{i}^{(S)}\}_{i=1}^{r}$ due to Lemma
\ref{lem:Dinv} in Appendix \ref{sec:Dinv}. Furthermore, $\beta$
Fisher information matrix can be calculated directly by Schur complement
of ${\rm Re}R+\beta\ii{\rm Im}R$ as follows. 
\begin{lem}
Let $D_{i}^{(\beta)}=(I+\beta\sqrt{-1}\D_{\rho_{\theta_{0}}})^{-1}(D_{i}^{(S)})$
be the $\beta$ logarithmic derivative with respect to the extended
SLD $D_{i}$ ($1\leq i\leq r$) at $\theta_{0}$ given in Corollary
\ref{cor:easy_holevo2}, and let
\begin{equation}
\tilde{J}_{\theta_{0}}^{(\beta)}=\left[\left\langle D_{i}^{(\beta)},D_{j}^{(\beta)}\right\rangle ^{(\beta)}\right]_{1\leq i,j\leq r}
\end{equation}
be the extended $\beta$ Fisher information matrix. Then $R$ given
in Corollary \ref{cor:easy_holevo2} satisfies
\begin{equation}
\tilde{J}_{\theta_{0}}^{(\beta)^{-1}}={\rm Re}R+\beta\ii{\rm Im}R.\label{eq:ex_invbJ}
\end{equation}
Further, the inverse $\beta$ Fisher information matrix $J_{\theta_{0}}^{(\beta)^{-1}}$can
be represented by $R$ as
\begin{equation}
J_{\theta_{0}}^{(\beta)^{-1}}=R_{1}^{(\beta)}-R_{2}^{(\beta)^{*}}R_{3}^{(\beta)^{-1}}R_{2}^{(\beta)},\label{eq:R_bLDFisher}
\end{equation}
where $R_{1}^{(\beta)}={\rm Re}R_{1}+\beta\ii{\rm Im}R_{1}$, $R_{2}^{(\beta)}={\rm Re}R_{2}+\beta\ii{\rm Im}R_{2}$,
$R_{3}^{(\beta)}={\rm Re}R_{3}+\beta\ii{\rm Im}R_{3}.$
\end{lem}

\begin{proof}
The proof of (\ref{eq:ex_invbJ}) is given in Lemma \ref{lem:Dinv}.
The proof of (\ref{eq:R_bLDFisher}) is immediate, because $J_{\theta_{0}}^{(\beta)}$
is the $(1,1)$ block of

\begin{equation}
\tilde{J}_{\theta_{0}}^{(\beta)}=\begin{pmatrix}R_{1}^{(\beta)} & R_{2}^{(\beta)^{*}}\\
R_{2}^{(\beta)} & R_{3}^{(\beta)}
\end{pmatrix}^{-1},
\end{equation}
and it is the same as the inverse of $\tilde{J}_{\theta_{0}}^{(\beta)^{-1}}/R_{3}^{(\beta)}=R_{1}^{(\beta)}-R_{2}^{(\beta)^{*}}R_{3}^{(\beta)^{-1}}R_{2}^{(\beta)}$,
where $\tilde{J}_{\theta_{0}}^{(\beta)^{-1}}/R_{3}^{(\beta)}$ is
the Schur complement given in Appendix \ref{sec:schur}. 
\end{proof}
From this lemma, a relation between the Holevo bound and $C_{\theta_{0},G}^{(\beta)}$
can be obtained directly.
\begin{lem}
For any $\beta\in[0,1]$,
\begin{equation}
C_{\theta_{0},G}^{(H)}\geq C_{\theta_{0},G}^{(\beta)}.
\end{equation}
\end{lem}

\begin{proof}
Let $Z^{(\beta)}(f):=(I,f^{\intercal})({\rm Re}R+\beta\ii{\rm Im}R)\begin{pmatrix}I\\
f
\end{pmatrix}.$ Then we see
\begin{align}
C_{\theta_{0},G}^{(H)} & =\min_{f}\left\{ \Tr GZ(f)+\Tr\left|\sqrt{G}{\rm Im}Z(f)\sqrt{G}\right|;\right.\\
 & \qquad f\text{ is an \ensuremath{(r-d)\times d} real matrix}\},\\
 & \geq\min_{f}\left\{ \Tr GZ^{(\beta)}(f)+\Tr\left|\sqrt{G}{\rm Im}Z^{(\beta)}(f)\sqrt{G}\right|;\right.\\
 & \qquad f\text{ is an \ensuremath{(r-d)\times d} real matrix}\}\\
 & \geq\min_{f}\left\{ \Tr GZ^{(\beta)}(f)+\Tr\left|\sqrt{G}{\rm Im}Z^{(\beta)}(f)\sqrt{G}\right|;\right.\\
 & \qquad f\text{ is an \ensuremath{(r-d)\times d} complex matrix}\}=C_{\theta_{0},G}^{(\beta)}.
\end{align}
The last equality is obtained from
\begin{align}
Z^{(\beta)}(f) & =R_{1}^{(\beta)}-R_{2}^{(\beta)^{*}}R_{3}^{(\beta)^{-1}}R_{2}^{(\beta)}+(f^{*}+R_{2}^{(\beta)^{*}}R_{3}^{(\beta)^{-1}})R_{3}^{(\beta)}(f+R_{3}^{(\beta)^{-1}}R_{2}^{(\beta)})\\
 & \geq R_{1}^{(\beta)}-R_{2}^{(\beta)^{*}}R_{3}^{(\beta)^{-1}}R_{2}^{(\beta)}=J_{\theta_{0}}^{(\beta)^{-1}}
\end{align}
and the minimum is achieved when $f=-R_{3}^{(\beta)^{-1}}R_{2}^{(\beta)}$.
\end{proof}

\textcolor{black}{For a numerical computation of the Holevo bound, it
was proposed to apply a linear semi-definite program\cite{sdp}. The
minimization problem given in Corollary \ref{cor:easy_holevo2} can
be rewritten to a linear semi-definite program:
\begin{equation}
{\rm minimize}_{f,V}\,\Tr GV
\end{equation}
\begin{equation}
\text{subject to }\begin{pmatrix}V & \begin{pmatrix}I_{d} & f^{\intercal}\end{pmatrix}\sqrt{R}\\
\sqrt{R}\begin{pmatrix}I_{d}\\
f
\end{pmatrix} & I_{r}
\end{pmatrix}\geq0,\label{eq:sdp2}
\end{equation}
where $V$ is a $d\times d$ real matrix, $I_{d}$ and $I_{r}$ are
the identity matrices of size $d$ and $r$, since we can see that
the inequality $V\geq Z(f)=\begin{pmatrix}I_{d} & f^{\intercal}\end{pmatrix}R\begin{pmatrix}I_{d}\\
f
\end{pmatrix}$ is equivalent to (\ref{eq:sdp2}) by considering the Schur complement
of (\ref{eq:sdp2}).}

\textcolor{black}{
The relationship between the bounds introduced in
this paper is
\begin{equation}
2C_{\theta_{0},G}^{(S)}\geq C_{\theta_{0},G}^{(H)}\geq\max_{0\leq\beta\leq1}C_{\theta_{0},G}^{(\beta)}\geq\max\{C_{\theta_{0},G}^{(S)},C_{\theta_{0},G}^{(R)}\}.
\end{equation}
In the first inequality, $2C_{\theta_{0},G}^{(S)}$ is known as an
upper bound\cite{upper1,upper2} of the Holevo bound. 
This inequality can be shown as follows. 
In Corollary \ref{cor:easy_holevo2},  it can be seen that
\begin{equation}
\Tr G\,{\rm Re}Z(f^{(S)})\leq C_{\theta_{0},G}^{(H)}\leq\Tr G\,{\rm Re}Z(f^{(S)})+\Tr\left|\sqrt{G}{\rm Im}Z(f^{(S)})\sqrt{G}\right|,\label{eq:fac2proof}
\end{equation}
where $f^{(S)}:=-({\rm Re}R_{3})^{-1}({\rm Re}R_{2})$. Because
\begin{equation}
{\rm Re}Z(f^{(S)})=J_{\theta_{0}}^{(S)^{-1}}\geq\ii{\rm Im}Z(f^{(S)}),
\end{equation}
we obtain the inequality
\begin{equation}
2C_{\theta_{0},G}^{(S)}\geq C_{\theta_{0},G}^{(H)}.\label{eq:fac2}
\end{equation}
}

\textcolor{black}{For any sufficiently smooth model $\left\{ \rho_{\theta};\theta\in\Theta\subset\R^{d}\right\} $,
it is known that a sequence of i.i.d. extension models $\left\{ \rho_{\theta_{0}+h/\sqrt{n}}^{\otimes n};h\in\R^{d}\right\} $
with a local parameter $h\in\R^{d}$ has a sequence of estimators
that achieves the Holevo bound $C_{\theta_{0},G}^{(H)}$ asymptotically
by using the theory of the quantum local asymptotic normality \cite{YFG,qlan2,guta}.
On the other hand, $C_{\theta_{0},G}^{(S)},C_{\theta_{0},G}^{(R)}$,
and $\max_{0\leq\beta\leq1}C_{\theta_{0},G}^{(\beta)}$ can not be
always achieved by considering i.i.d. extension. Therefore, these
bounds are informative only when they are consistent with the Holevo
bound. It can be obviously seen from Theorem \ref{thm:rewrite_holevo}
that $C_{\theta_{0},G}^{(H)}=C_{\theta_{0},G}^{(R)}$ if SLDs are
$\D_{\theta_{0}}$ invariant. It can be also seen from (\ref{eq:fac2proof})
that $C_{\theta_{0},G}^{(H)}=C_{\theta_{0},G}^{(S)}$ if and only
if ${\rm Im}Z(f^{(S)})=0$. Since the maximum logarithmic derivative
bound $\max_{0\leq\beta\leq1}C_{\theta_{0},G}^{(\beta)}$ is larger
than the SLD bound and the RLD bound, if the Holevo bound $C_{\theta_{0},G}^{(H)}$
is equal to the SLD bound or RLD bound, $C_{\theta_{0},G}^{(H)}=\max_{0\leq\beta\leq1}C_{\theta_{0},G}^{(\beta)}$
also holds. In Section \ref{sec:codimension}, we provide another
case of satisfying $C_{\theta_{0},G}^{(H)}=\max_{0\leq\beta\leq1}C_{\theta_{0},G}^{(\beta)}$
that is different from SLD or RLD bounds and has an explicit solution.
Further, we give examples of models that can achieve $\max_{0\leq\beta\leq1}C_{\theta_{0},G}^{(\beta)}$
(see Example \ref{exa:gauss} and \ref{exa:pure}).}

\textcolor{black}{When $\rho_{\theta_{0}}$ is not strictly positive,
the $\beta$ logarithmic derivatives $\{L_{i}^{(\beta)}\}_{i=1}^{d}$
that satisfy (\ref{eq:Lbeta_def}) for $-1<\beta<1$ can be defined
on the quotient space $\B(\H)/\sim_{\rho_{\theta_{0}}}$ with respect
to an equivalence relation defined by 
\begin{equation}
A\sim_{\rho_{\theta_{0}}}B\Leftrightarrow A-B\in{\rm Ker}\bL_{\rho_{\theta_{0}}}\cap{\rm Ker}\bR_{\rho_{\theta_{0}}}.
\end{equation}
The inner product $\left\langle \cdot,\cdot\right\rangle ^{(\beta)}$
on $\B(\H)/\sim_{\rho_{\theta_{0}}}$ and the $\beta$ Fisher information
matrix $J_{\theta_{0}}^{(\beta)}$ can be also defined by 
(\ref{eq:beta_inner})
and 
(\ref{eq:beta_fisher}). 
The commutation operator 
$\D_{\rho_{\theta_{0}}}$
is defined by (\ref{eq:D_def}) as a super operator on $\B(\H)/\sim_{\rho_{\theta_{0}}}$.
For $\beta=1$, RLDs $\{L_{i}^{(R)}\}_{i=1}^{d}$ cannot be defined,
however $\left\langle \cdot,\cdot\right\rangle ^{(1)}$ can be defined
as a pre-inner product on $\B(\H)/\sim_{\rho_{\theta_{0}}}$ and (\ref{eq:bLD_SLD_inner})
is valid. 
Theorem \ref{thm:rewrite_holevo} and Corollary \ref{cor:easy_holevo2}
also holds in a similar way by dealing with $\B(\H)/\sim_{\rho_{\theta_{0}}}$
instead of $\B(\H)$. }

\section{\label{sec:codimension}$\protect\D_{\rho_{\theta_{0}}}$ invariant
extension with one dimension}

In Corollary \ref{cor:easy_holevo2}, if $\{D_{j}\}_{j=d+1}^{r}$
are orthogonal to $\T={\rm span}_{\R}\{D_{i}\}_{i=1}^{d}$ with respect
to the inner product $\left\langle \cdot,\cdot\right\rangle ^{(0)}$,
${\rm Re}R_{2}=0$. Further, if $r=d+1$ and $\left\langle D_{r},D_{r}\right\rangle ^{(0)}=1$,
$R$ can take form of
\begin{equation}
R=\begin{pmatrix}A & \sqrt{-1}\ket b\\
-\sqrt{-1}\bra b & 1
\end{pmatrix}\label{eq:D-1model}
\end{equation}
with a real vector $\ket b\in\R^{d}$. In this case,
\begin{equation}
J_{\theta_{0}}^{(\beta)^{-1}}={\rm Re}A+\beta\ii{\rm Im}A-\beta^{2}\ket b\bra b
\end{equation}
due to (\ref{eq:R_bLDFisher}), and
\begin{equation}
C_{\theta_{0},G}^{(\beta)}=\Tr G{\rm Re}A+\beta\Tr\left|\sqrt{G}{\rm Im}A\sqrt{G}\right|-\beta^{2}\bra bG\ket b.
\label{eq:bLD_Ab}
\end{equation}
Therefore $A$ and $\ket b\bra b$ can be expressed by $J_{\theta_{0}}^{(R)^{-1}}$
and $J_{\theta_{0}}^{(S)^{-1}}$ as
\begin{equation}
A=J_{\theta_{0}}^{(S)^{-1}}+\sqrt{-1}{\rm Im}(J_{\theta_{0}}^{(R)^{-1}})\label{eq:A_relation}
\end{equation}
\begin{equation}
\ket b\bra b=J_{\theta_{0}}^{(S)^{-1}}-{\rm Re}(J_{\theta_{0}}^{(R)^{-1}}).\label{eq:b_relation}
\end{equation}
Let us calculate the maximum logarithmic derivative bound
\begin{equation}
\max_{0\leq\beta\leq1}C_{\theta_{0},G}^{(\beta)}=\max_{0\leq\beta\leq1}\Tr G{\rm Re}A+\beta\Tr\left|\sqrt{G}{\rm Im}A\sqrt{G}\right|-\beta^{2}\bra bG\ket b.\label{eq:max_program}
\end{equation}
If $\ket b\not=0$, the quadratic function
\begin{equation}
g_{1}:\beta\mapsto\Tr G{\rm Re}A+\beta\Tr\left|\sqrt{G}{\rm Im}A\sqrt{G}\right|-\beta^{2}\bra bG\ket b
\end{equation}
is maximized at
\begin{equation}
\beta=\frac{\Tr\left|\sqrt{G}{\rm Im}A\sqrt{G}\right|}{2\bra bG\ket b}>0.
\end{equation}
If $\frac{\Tr\left|\sqrt{G}{\rm Im}A\sqrt{G}\right|}{2\bra bG\ket b}\geq1$,
$C_{\theta_{0},G}^{(\beta)}$ is maximized at $\beta=1$, thus
\begin{equation}
\max_{0\leq\beta\leq1}C_{\theta_{0},G}^{(\beta)}=g_{1}(1)=\Tr G{\rm Re}A+\Tr\left|\sqrt{G}{\rm Im}A\sqrt{G}\right|-\bra bG\ket b.\label{eq:max1}
\end{equation}
If $\frac{\Tr\left|\sqrt{G}{\rm Im}A\sqrt{G}\right|}{2\bra bG\ket b}<1$,
$C_{\theta_{0},G}^{(\beta)}$ is maximized at $\beta=\frac{\Tr\left|\sqrt{G}{\rm Im}A\sqrt{G}\right|}{2\bra bG\ket b}$,
thus
\begin{equation}
\max_{0\leq\beta\leq1}C_{\theta_{0},G}^{(\beta)}=g_{1}(\frac{\Tr\left|\sqrt{G}{\rm Im}A\sqrt{G}\right|}{2\bra bG\ket b})=\Tr G{\rm Re}A+\frac{\left\{ \Tr\left|\sqrt{G}{\rm Im}A\sqrt{G}\right|\right\} ^{2}}{4\bra bG\ket b}.\label{eq:max2}
\end{equation}
When $\ket b=0$, $g_{1}$ is a linear function and $C_{\theta_{0},G}^{(\beta)}$
is maximized at $\beta=1$, so
\begin{equation}
\max_{0\leq\beta\leq1}C_{\theta_{0},G}^{(\beta)}=g_{1}(1)=\Tr G{\rm Re}A+\Tr\left|\sqrt{G}{\rm Im}A\sqrt{G}\right|.
\end{equation}

Collecting these result, we have the following theorem.
\begin{thm}
\label{thm:mld_1}When the model has $d+1$ dimensional $\D_{\rho_{\theta_{0}}}$
invariant extended SLD tangent space, the maximum logarithmic derivative
bound is
\begin{equation}
\max_{0\leq\beta\leq1}C_{\theta_{0},G}^{(\beta)}=\begin{cases}
C_{\theta_{0},G}^{(1)} & \text{if }\hat{\beta}\geq1,\\
C_{\theta_{0},G}^{(\hat{\beta})} & \text{\text{otherwise}},
\end{cases}
\end{equation}
where
\begin{equation}
\hat{\beta}=\begin{cases}
\frac{\Tr\left|\sqrt{G}{\rm Im}J_{\theta_{0}}^{(R)^{-1}}\sqrt{G}\right|}{2\Tr G\left\{ J_{\theta_{0}}^{(S)^{-1}}-{\rm Re}(J_{\theta_{0}}^{(R)^{-1}})\right\} } & \text{if }J_{\theta_{0}}^{(S)^{-1}}\not={\rm Re}(J_{\theta_{0}}^{(R)^{-1}}),\\
\infty & \text{\text{otherwise}}.
\end{cases}
\end{equation}
\end{thm}

In general, even when $R$ can take form of (\ref{eq:D-1model}),
the equality of
\begin{equation}
C_{\theta_{0},G}^{(H)}\geq\max_{0\leq\beta\leq1}C_{\theta_{0},G}^{(\beta)}
\end{equation}
is not always achieved. However, when $d=2$, these two bounds are
consistent.
\begin{thm}
\label{thm:holevo_mld}When $d=2$ and the model has three dimensional
$\D_{\rho_{\theta_{0}}}$ invariant extended SLD tangent space, the
Holevo bound $C_{\theta_{0},G}^{(H)}$ is the same as $\max_{0\leq\beta\leq1}C_{\theta_{0},G}^{(\beta)}$
which is given explicitly in Theorem \ref{thm:mld_1}. 
\end{thm}

\begin{proof}
By using Corollary \ref{cor:easy_holevo2}, the Holevo bound is
\begin{eqnarray}
C_{\theta_{0},G}^{(H)} & = & \min_{f}\left\{ \Tr G\,Z(f)+\Tr\left|\sqrt{G}{\rm Im}Z(f)\sqrt{G}\right|\right\} ,\label{eq:min_f_2}
\end{eqnarray}
where
\begin{equation}
Z(f)={\rm Re}A+\ket f\bra f+\sqrt{-1}\left({\rm Im}A+\ket b\bra f-\ket f\bra b\right).
\end{equation}
Letting
\begin{equation}
\ket{\hat{b}}:=\sqrt{G}\ket b,
\end{equation}
\begin{equation}
\ket{\hat{f}}:=\sqrt{G}\ket f,
\end{equation}
and
\begin{equation}
\hat{A}:=\sqrt{G}{\rm Im}A\sqrt{G}=a\begin{pmatrix}0 & -1\\
1 & 0
\end{pmatrix}
\end{equation}
with $a\in\R$, then
\begin{eqnarray}
C_{\theta_{0},G}^{(H)} & = & \min_{\hat{f}}\left\{ \Tr G\,{\rm Re}A+\braket{\hat{f}}{\hat{f}}+\Tr\left|\hat{A}+\ket{\hat{b}}\bra{\hat{f}}-\ket{\hat{f}}\bra{\hat{b}}\right|\right\} \\
 & = & \min_{\hat{f}}\left\{ \Tr G\,{\rm Re}A+\braket{\hat{f}}{\hat{f}}+\Tr\left|(a+\hat{b}_{2}\hat{f}_{1}-\hat{b}_{1}\hat{f}_{2})\begin{pmatrix}0 & -1\\
1 & 0
\end{pmatrix}\right|\right\} \\
 & = & \min_{\hat{f}}\left\{ \Tr G\,{\rm Re}A+\braket{\hat{f}}{\hat{f}}+2\left|a+\hat{b}_{2}\hat{f}_{1}-\hat{b}_{1}\hat{f}_{2}\right|\right\} .
\end{eqnarray}
By representing $\ket{\hat{f}}$ as $\ket{\hat{f}}=s\begin{pmatrix}\hat{b}_{1}\\
\hat{b}_{2}
\end{pmatrix}+t\begin{pmatrix}\hat{b}_{2}\\
-\hat{b}_{1}
\end{pmatrix}$ with $s,t\in\R$, we have
\begin{align}
C_{\theta_{0},G}^{(H)} & =\min_{s,t}\left\{ \Tr G\,{\rm Re}A+\braket{\hat{b}}{\hat{b}}s^{2}+\braket{\hat{b}}{\hat{b}}t^{2}+2\left|a+\braket{\hat{b}}{\hat{b}}t\right|\right\}\\
 & =\min_{t}\left\{ \Tr G\,{\rm Re}A+\braket{\hat{b}}{\hat{b}}t^{2}+2\left|a+\braket{\hat{b}}{\hat{b}}t\right|\right\}\\
 & =\min_{t}\left\{ \Tr G\,{\rm Re}A+\braket{\hat{b}}{\hat{b}}t^{2}+2\left|\left|a\right|+\braket{\hat{b}}{\hat{b}}t\right|\right\}\\
 & \leq\min_{\left|a\right|+\braket{\hat{b}}{\hat{b}}t\geq0}\left\{ \Tr G\,{\rm Re}A+\braket{\hat{b}}{\hat{b}}t^{2}+2\left|a\right|+2\braket{\hat{b}}{\hat{b}}t\right\} .\label{eq:min_program}
\end{align}
Let
\begin{equation}
g_{2}(t)=\Tr G\,{\rm Re}A+\braket{\hat{b}}{\hat{b}}t^{2}+2\left|a\right|+2\braket{\hat{b}}{\hat{b}}t.
\end{equation}
Because of (\ref{eq:max1}), (\ref{eq:max2}), and $C_{\theta_{0},G}^{(H)}\geq\max_{0\leq\beta\leq1}C_{\theta_{0},G}^{(\beta)}$
, we have 
\begin{equation}
g_{2}(t)\geq\begin{cases}
\Tr G\,{\rm Re}A+2\left|a\right|-\braket{\hat{b}}{\hat{b}}, & \text{if }\left|a\right|\geq\braket{\hat{b}}{\hat{b}},\\
\Tr G\,{\rm Re}A+\frac{a^{2}}{\braket{\hat{b}}{\hat{b}}} & \text{\text{otherwise}},
\end{cases}
\end{equation}
where the equality is achieved at $t=\max\{-\frac{\left|a\right|}{\braket{\hat{b}}{\hat{b}}},-1\}$. 
\end{proof}
Let us consider the Lagrangian duality of the quadratic programming
(\ref{eq:min_program}). The Lagrangian function is
\begin{align}
\mathscr{L}(t,\lambda) & =\Tr G\,{\rm Re}A+2\left|a\right|+2\braket{\hat{b}}{\hat{b}}t+\braket{\hat{b}}{\hat{b}}t^{2}-\lambda(\left|a\right|+\braket{\hat{b}}{\hat{b}}t)\\
 & =\Tr G\,{\rm Re}A+2\left|a\right|-\lambda\left|a\right|+\braket{\hat{b}}{\hat{b}}((2-\lambda)t+t^{2}).
\end{align}
For any fixed $\lambda\in\R,$ 
$\mathscr{L}(t,\lambda)$ is minimized
at $t=\frac{\lambda-2}{2}$, and the Lagrangian dual function is
\begin{align}
g(\lambda) & =\min_{t}\mathscr{L}(t,\lambda)=\mathscr{L}(\frac{\lambda-2}{2},\lambda)=\Tr G\,{\rm Re}A-(\lambda-2)\left|a\right|-\braket{\hat{b}}{\hat{b}}\frac{(\lambda-2)^{2}}{4}\\
 & =g_{1}(\frac{2-\lambda}{2}).
\end{align}
Hence the Lagrangian dual programming is
\begin{equation}
\max_{\lambda\geq0}g_{1}(\frac{2-\lambda}{2}).
\end{equation}
The solution of this maximization is the same as (\ref{eq:max_program}).
It is known that in quadratic programming the Lagrangian duality problem
has the same solution. This is the reason why the two bounds coincide.

\textcolor{black}{
The optimal observables $B_{1},B_{2}$ for the minimization
of (\ref{eq:holevo_bound2}) to define Holevo bound can be described
by $\beta^{*}:={\rm argmax_{0\leq\beta\leq1}}C_{\theta_{0},G}^{(\beta)}$
given in Theorem \ref{thm:mld_1} as follows. 
From the proof of Theorem \ref{thm:holevo_mld}, 
we see that the minimization of 
(\ref{eq:min_f_2})
is achieved when 
\begin{equation}
\ket f=\ket{f^{(\beta^{*})}}:=\beta^{*}\sqrt{G}^{-1}\begin{pmatrix}0 & -1\\
1 & 0
\end{pmatrix}\sqrt{G}\ket b.\label{eq:beta2f}
\end{equation}
This means the minimization of 
(\ref{eq:min_F})
is achieved when
$F=F^{(\beta^{*})}:=\begin{pmatrix}I\\
\bra{f^{(\beta^{*})}}
\end{pmatrix}({\rm Re}\Sigma)^{-1}$, and the minimization of (\ref{eq:holevo_bound2}) is achieved when
\begin{equation}
B_{i}=B_{i}^{(\beta^{*})}:=\sum_{j=1}^{3}F_{ji}^{(\beta^{*})}D_{j}\qquad(i=1,2).\label{eq:beta2B}
\end{equation}
}

When $\dim\H=2$ and $d=2$, any model $\S$ has two SLDs $L_{1}$
and $L_{2}$ at any point $\theta_{0}$, and $\tilde{\T}=\{X\in\B_{h}(\H);\Tr\rho_{\theta_{0}}X=0\}\supset{\rm span}_{\R}\{L_{1},L_{2}\}$
is $\D_{\rho_{\theta_{0}}}$ invariant three dimensional space. Therefore,
$R$ takes the form of (\ref{eq:D-1model}), thus Theorem \ref{thm:mld_1}
and Theorem \ref{thm:holevo_mld} are applicable. 
This is the essential
reason why the Holevo bound can be expressed by (\ref{eq:suzuki_bound_ori}). 

\textcolor{black}{
More generally, if a two-dimensional smooth parametric
family $\{\sigma_{\xi};\xi\in\Xi\subset\R^{2}\}$ of density operators
on a Hilbert space $\H$ with an open set $\Xi\subset\R^{2}$ is $\D_{\xi}$
invariant, a three-dimensional smooth parametric family $\{\tilde{\rho}_{(\xi,\eta)}=\eta\sigma_{\xi}+(1-\eta)\frac{1}{\dim\H}I;\,\xi\in\Xi\subset\R^{2},0<\eta<1\}$
is also $\D_{(\xi,\eta)}$ invariant. 
Therefore, Theorem \ref{thm:mld_1}
and \ref{thm:holevo_mld} are applicable for any two-dimensional submodel
of $\left\{ \tilde{\rho}_{(\xi,\eta)}\right\} _{(\xi,\eta)}$, and
the maximum logarithmic derivative bound and the Holevo bound can
be calculated explicitly. We show examples below. }

\begin{example}
\label{exa:dim2}Let
\begin{equation}
\left\{ \rho_{\theta}=a(1-|\theta|)\frac{1}{2}(\theta^{1}\sigma_{1}+\theta^{2}\sigma_{2}+\sqrt{1-|\theta|^{2}}\sigma_{3})+(1-a(1-|\theta|))\frac{I}{2};|\theta|<1\right\} 
\end{equation}
be a family of density operators on $\H=\C^{2}$ parameterized by
$\theta=(\theta^{1},\theta^{2})$ with fixed $0<a<1$, where $\sigma_{1},\sigma_{2},\sigma_{3}$
are Pauli matrices. Let $D_{1}=\partial_{1}\rho_{\theta}$, $D_{2}=\partial_{2}\rho_{\theta}$,
$D_{3}=\rho_{\theta}-\frac{I}{2}$. A linear space of observables
${\rm span}_{\R}\left\{ D_{i}\right\} _{i=1}^{d}=\{X\in\B_{h}(\H);\Tr X=0\}$
is $\D_{\theta}$ invariant at any $\theta$. The extended RLD Fisher
information matrix is calculated by $\tilde{J}_{\theta,ij}^{(R)}=\Tr D_{i}\rho_{\theta}^{-1}D_{j}$
and its inverse is
\begin{equation}
\tilde{J}_{\theta,ij}^{(R)^{-1}}=\frac{1}{a^{2}(1-r)}\begin{pmatrix}1+r & -\ii a\sqrt{1-r^{2}} & \frac{1+r}{1-r}\\
\ii a\sqrt{1-r^{2}} & \frac{1}{(1-r)} & \frac{\ii a(1+r)}{\sqrt{1-r^{2}}}\\
\frac{1+r}{1-r} & -\frac{\ii a(1+r)}{\sqrt{1-r^{2}}} & a^{2}(x-1)+\frac{2}{(1-r)^{2}}
\end{pmatrix},
\end{equation}
at $\theta=(r,0)$ with $0\leq r<1$. The inverse $\beta$ Fisher
information matrix $J_{\theta}^{(\beta)^{-1}}$ is the Schur complement
of ${\rm Re}\tilde{J}_{\theta,ij}^{(R)^{-1}}+\beta\ii{\rm Im}\tilde{J}_{\theta,ij}^{(R)^{-1}}$
due to (\ref{eq:R_bLDFisher}). Let us consider lower bounds of $\Tr GV_{\theta_{0}}[M,\hat{\theta}]$
with a SLD weight $G=J_{\theta}^{(S)}$. The $\beta$ bound is
\begin{equation}
C_{\theta_{0},G}^{(\beta)}=2+\frac{2a\sqrt{(1-a^{2}(1-r)^{2})(2-a^{2}(1-r)^{3})(1-r)^{3}}}{2-a^{2}(1-r)^{3}}\beta-\frac{a^{2}(1-r)^{2}(1+r)}{2-a^{2}(1-r)^{3}}\beta^{2}.
\end{equation}
By using Theorem \ref{thm:mld_1}, we see that the maximum of $C_{\theta_{0},G}^{(\beta)}$
is achieved by
\begin{equation}
\beta=\min\left\{ 1,\frac{1}{a(1+r)}\sqrt{\frac{(1-a^{2}(1-r)^{2})(2-a^{2}(1-r)^{3})}{1-r}}\right\} .
\end{equation}
In Fig \ref{fig:beta_plot}(left), the behavior of the optimal $\beta$
is plotted as a function of $r$ when $a=0.95$. Due to Theorem \ref{thm:holevo_mld},
$\max_{0\leq\beta\leq1}C_{\theta_{0},G}^{(\beta)}$ is the same as
the Holevo bound $C_{\theta_{0},G}^{(H)}$. This result illustrates
a principle behind the explicit expression of the Holevo bound (\ref{eq:suzuki_bound_ori}).
\end{example}

\begin{example}
\label{exa:dim4}Here we show an example of the case when $\dim\H>2$.
Let
\begin{equation}
\left\{ \rho_{\theta}=a(1-|\theta|)\left\{ \frac{1}{2}(\theta^{1}\sigma_{1}+\theta^{2}\sigma_{2}+\sqrt{1-|\theta|^{2}}\sigma_{3})\right\} ^{\otimes2}+(1-a(1-|\theta|))\frac{I}{4};|\theta|<1\right\} 
\end{equation}
be a family of density operators on $\H=\C^{4}$ parameterized by
$\theta=(\theta^{1},\theta^{2})$ with fixed $0<a<1$. Let $D_{1}=\partial_{1}\rho_{\theta}$,
$D_{2}=\partial_{2}\rho_{\theta}$, $D_{3}=\rho_{\theta}-\frac{I}{4}$.
A linear space of observables ${\rm span}_{\R}\left\{ D_{i}\right\} _{i=1}^{d}$
is $\D_{\theta}$ invariant at any $\theta$. The extended RLD Fisher
information matrix is calculated by $\tilde{J}_{\theta,ij}^{(R)}=\Tr D_{i}\rho_{\theta}^{-1}D_{j}$.
By the similar calculation as in Example \ref{exa:dim2}, we see that
the maximum of $C_{\theta_{0},G}^{(\beta)}$ is achieved by
\begin{equation}
\beta=\min\left\{ 1,\frac{1}{3a(1+r)}\sqrt{\frac{(1-a(1-r))(1+3a(1-r))(7-r-a(1-r)(-11+12a(1-r)^{2}+5r))}{1-r}}\right\} 
\end{equation}
for a SLD weight $G=J_{\theta}^{(S)}$ at $\theta=(r,0)$ with $0\leq r<1$.
In Fig \ref{fig:beta_plot}(right), the behavior of the optimal $\beta$
is plotted as a function of $r$ when $a=0.95$. Due to Theorem \ref{thm:holevo_mld},
$\max_{0\leq\beta\leq1}C_{\theta_{0},G}^{(\beta)}$ is the same as
the Holevo bound $C_{\theta_{0},G}^{(H)}$. This example is not included
in the result of (\ref{eq:suzuki_bound_ori}) for $\dim\H=2$.

\begin{figure}
\begin{centering}
\includegraphics{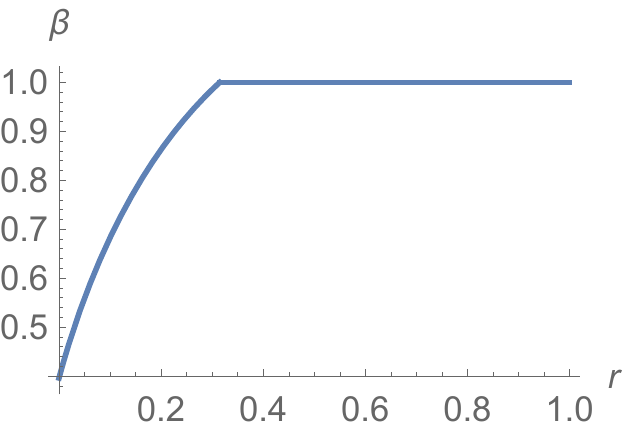}\includegraphics{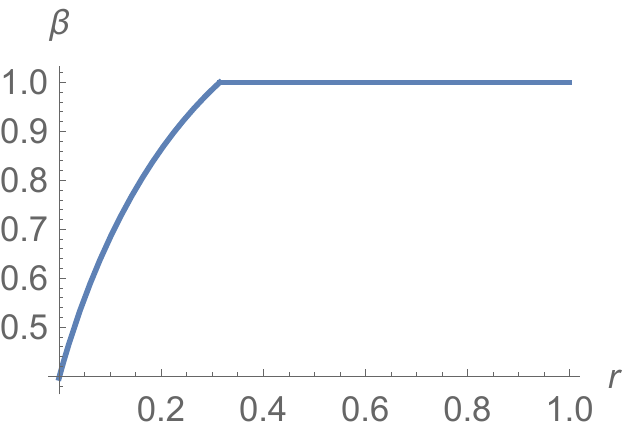}
\par\end{centering}
\caption{\label{fig:beta_plot}The behavior of the optimal $\beta$ as a function
of $r$ for Example \ref{exa:dim2} (left) and Example \ref{exa:dim4}
(right) with $a=0.95$.}
\end{figure}
\end{example}

\textcolor{black}{Next, let us show examples of models that can achieve
the maximum logarithmic derivative bounds. It is known that the Holevo
bounds can be achieved for quantum Gaussian shift models\cite{holevo}
and pure states models\cite{matsumoto_pure}. The Holevo bounds can
also be achieved as SLD bounds for models that have commutative SLDs.
We can derive the similar property by combining the above models that
have the achievable Holevo bounds. The following examples show models
of a tensor product of a one-dimensional model and quantum Gaussian
shift models or pure states model that have achievable maximum logarithmic
derivative bounds.}
\begin{example}
\textcolor{black}{\label{exa:gauss}Let
\begin{equation}
\left\{ \sigma_{\eta}^{(1)};\text{\ensuremath{\eta_{1}<\eta<\eta_{2}}}\right\} 
\end{equation}
be any one-dimensional family of density operators on a Hilbert space
$\H_{1}$ parameterized by $\eta\in\R$, and let
\begin{equation}
\left\{ \sigma_{\xi}^{(2)};\,\xi\in\R^{2}\right\} 
\end{equation}
be a two-dimensional family of quantum Gaussian states, where $\sigma_{\xi}^{(2)}$
is a quantum Gaussian state\cite{holevo,YFG} represented on a Hilbert
space $\H_{2}$ defined by a characteristic function
\begin{equation}
\varphi_{\xi}^{(2)}(\zeta)=\Tr\sigma_{\xi}^{(2)}e^{\ii\zeta^{i}X_{i}}=e^{-\frac{1}{2}s^{2}|\xi|^{2}+\ii\xi^{\intercal}\zeta}\qquad(\zeta\in\R^{2})
\end{equation}
with $s\geq1$ and canonical observables $X_{1},X_{2}$ such that
\begin{equation}
[X_{1},X_{2}]=2\ii I.
\end{equation}
Let us consider a three-dimensional quantum statistical model
\begin{equation}
\left\{ \tilde{\rho}_{(\eta,\xi)}=\sigma_{\eta}^{(1)}\otimes\sigma_{\xi}^{(2)};\text{\ensuremath{\eta_{1}<\eta<\eta_{2}},\,\ensuremath{\xi\in\R^{2}}}\right\} .
\end{equation}
Since it is known that $\tilde{D}_{i}^{(2)}=\frac{1}{s^{2}}X_{i}$
($i=1,2$) are the SLDs of $\sigma_{\xi}^{(2)}$ and their tangent
space is $\D_{\xi}$ invariant, the SLD tangent space of this three-dimensional
model is $\D_{(\eta,\xi)}$ invariant at every $(\eta,\xi)$. Therefore,
Theorem \ref{thm:mld_1} and \ref{thm:holevo_mld} are applicable
for any two-dimensional submodel $\left\{ \rho_{\theta};\,\theta\in\Theta\subset\R^{2}\right\} $
of $\left\{ \tilde{\rho}_{(\eta,\xi)}\right\} _{(\eta,\xi)}$, and
the maximum logarithmic derivative bound and the Holevo bound can
be calculated explicitly. Furthermore, we can show that the maximum
logarithmic derivative bound can be achieved. Let $D_{1},D_{2}$ be
SLDs of $\rho_{\theta}$ at $\theta=\theta_{0}$, let $D_{3}$ and
$R$ be an observable and a $3\times3$ matrix obtained in the same
way as (\ref{eq:D-1model}). Note that $D_{1},D_{2},D_{3}$ are in
${\rm span}_{\R}\{\tilde{D}^{(1)}\otimes I,I\otimes\tilde{D}_{1}^{(2)},I\otimes\tilde{D}_{2}^{(2)}\}$,
where $\tilde{D}^{(1)}$ is the SLD of $\sigma_{\eta}^{(1)}$ and
$\tilde{D}_{i}^{(2)}=\frac{1}{s^{2}}X_{i}$ ($i=1,2$) are the SLDs
of $\sigma_{\xi}^{(2)}$. Due to Theorem \ref{thm:mld_1}, the maximum
of $C_{\theta_{0},G}^{(\beta)}$ is achieved when $\beta=\beta^{*}:=\min\{1,\frac{\Tr\left|\sqrt{G}{\rm Im}A\sqrt{G}\right|}{2\bra bG\ket b}\}$
for any weight matrix $G$. By using (\ref{eq:beta2B}), it can be
seen that the minimization of (\ref{eq:holevo_bound2}) is achieved
when $B_{i}=B_{i}^{(\beta^{*})}:=\sum_{j=1}^{3}F_{ij}^{(\beta^{*})}D_{j}$
($i=1,2$). Note that $B_{1},B_{2}$ satisfy a commutation relation
\begin{equation}
[B_{1},B_{2}]=-2\ii{\rm Im}Z(B)_{12}I.
\end{equation}
Let $\sigma^{(3)}$ be another ancilla Gaussian states defined by
a characteristic function
\begin{equation}
\varphi^{(3)}(\zeta)=\Tr\sigma^{(3)}e^{\ii\zeta^{i}Y_{i}}=e^{-\frac{1}{2}\zeta^{\intercal}V^{(3)}\zeta}\qquad(\zeta\in\R^{2})
\end{equation}
with canonical observables $Y_{1},Y_{2}$ such that
\begin{equation}
[Y_{1},Y_{2}]=2\ii{\rm Im}Z(B)_{12}I
\end{equation}
and a real positive matrix
\begin{equation}
V^{(3)}=\sqrt{G}^{-1}\left|\sqrt{G}{\rm Im}Z(B)\sqrt{G}\right|\sqrt{G}^{-1}.
\end{equation}
It can be seen that two observables $\hat{B}_{i}:=\theta_{0}^{i}+B_{i}\otimes I+I\otimes Y_{i}$
($i=1,2$) can be measured simultaneously because they are commutative.
Further, these observables satisfy locally unbiased conditions and
achieve the Holevo bound, i.e.,
\begin{align}
\Tr(\rho_{\theta}\otimes\sigma^{(3)})\hat{B}_{i} & =\theta_{0}^{i}\qquad(1\leq i\leq2)\\
\Tr(\partial_{i}\rho_{\theta}\otimes\sigma^{(3)})\hat{B}_{j} & =\delta_{ij}\qquad(1\leq i,j\leq2)\\
\Tr(\rho_{\theta}\otimes\sigma^{(3)})\left(\hat{B}_{i}-\theta_{0}^{i}\right)\left(\hat{B}_{j}-\theta_{0}^{j}\right) & =({\rm Re}Z(B)+V^{(3)})_{ij}\qquad(1\leq i,j\leq2).
\end{align}
}
\end{example}

\begin{example}
\textcolor{black}{\label{exa:pure}Let
\begin{equation}
\left\{ \sigma_{\eta}^{(1)};\text{\ensuremath{\eta_{1}<\eta<\eta_{2}}}\right\} 
\end{equation}
be any one-dimensional family of density operators on a Hilbert space
$\H_{1}$ parameterized by $\eta\in\R$, and let
\begin{equation}
\left\{ \sigma_{\xi}^{(2)}=\ket{\psi_{\xi}}\bra{\psi_{\xi}};\,\xi\in\Xi\subset\R^{2}\right\} 
\end{equation}
be a two-dimensional family of pure states on a Hilbert space $\H_{2}$
with an open set $\Xi\subset\R^{2}$. Let us consider a three-dimensional
quantum statistical model
\begin{equation}
\left\{ \tilde{\rho}_{(\eta,\xi)}=\sigma_{\eta}^{(1)}\otimes\sigma_{\xi}^{(2)};\text{\ensuremath{\eta_{1}<\eta<\eta_{2}},\,\ensuremath{\xi\in\Xi\subset\R^{2}}}\right\} .
\end{equation}
Suppose ${\rm span}_{\R}\left\{ \partial_{i}\sigma_{\xi}^{(2)}=\ket{\partial_{i}\psi_{\xi_{0}}}\bra{\psi_{\xi_{0}}}+\ket{\psi_{\xi_{0}}}\bra{\partial_{i}\psi_{\xi_{0}}}\right\} _{i=1}^{2}$
is $\D_{\xi_{0}}$ invariant at a fixed point $\xi_{0}$. It can be
seen that $\D_{\xi_{0}}$ invariance for $\left\{ \sigma_{\xi}^{(2)}\right\} _{\xi}$
is equivalent to $\ket{\partial_{2}\psi_{\xi_{0}}}\in{\rm span}_{\R}\left\{ \ket{\partial_{1}\psi_{\xi_{0}}},\ii\ket{\partial_{1}\psi_{\xi_{0}}}\right\} $.
Since this three-dimensional model is also $\D_{(\eta_{0},\xi_{0})}$
invariant at a fixed point $(\eta_{0},\xi_{0})$, Theorem \ref{thm:mld_1}
and \ref{thm:holevo_mld} are applicable for two-dimensional submodel
$\left\{ \rho_{\theta};\,\theta\in\Theta\subset\R^{2}\right\} $ of
$\left\{ \tilde{\rho}_{(\xi,\eta)}\right\} _{(\xi,\eta)}$ at $\theta_{0}$
such that $\rho_{\theta_{0}}=\tilde{\rho}_{(\eta_{0},\xi_{0})}$,
and the maximum logarithmic derivative bound and the Holevo bound
can be calculated explicitly. Furthermore, we can show that the maximum
logarithmic derivative bound can be achieved. Let $D_{1},D_{2}$ be
SLDs of $\rho_{\theta}$ at $\theta=\theta_{0}$, let $D_{3}$ and
$R$ be an observable and a $3\times3$ matrix obtained in the same
way as (\ref{eq:D-1model}). Due to Theorem \ref{thm:mld_1}, the
maximum of $C_{\theta_{0},G}^{(\beta)}$ is achieved when $\beta=\beta^{*}:=\min\{1,\frac{\Tr\left|\sqrt{G}{\rm Im}A\sqrt{G}\right|}{2\bra bG\ket b}\}$
for any weight matrix $G$. By using (\ref{eq:beta2B}), it can be
seen that the minimization of (\ref{eq:holevo_bound2}) is achieved
when $B_{i}=B_{i}^{(\beta^{*})}:=\sum_{j=1}^{3}F_{ji}^{(\beta^{*})}D_{j}$
($i=1,2$). Because $D_{1},D_{2},D_{3}$ are in ${\rm span}_{\R}\{\tilde{D}^{(1)}\otimes I,I\otimes\tilde{D}_{1}^{(2)},I\otimes\tilde{D}_{2}^{(2)}\}$,
where $\tilde{D}^{(1)}$ is the SLD of $\sigma_{\eta}^{(1)}$ and
$\tilde{D}_{i}^{(2)}=2\partial_{i}\sigma_{\xi_{0}}^{(2)}$ ($i=1,2$)
are the SLDs of $\sigma_{\xi}^{(2)}$, there exist $1\times2$ and
$2\times2$ real matrix $F^{(1)}$ and $F^{(2)}$ such that $B_{i}^{(\beta^{*})}=F_{1i}^{(1)}\tilde{D}^{(1)}\otimes I+\sum_{j=1}^{2}F_{ji}^{(2)}I\otimes\tilde{D}_{j}^{(2)}$.
Let $B_{i}^{(1)}:=F_{1i}^{(1)}\tilde{D}^{(1)}\in\B(\H_{1})$ and $B_{i}^{(2)}:=\sum_{j=1}^{2}F_{ji}^{(2)}\tilde{D}_{j}^{(2)}\in\B(\H_{2})$,
and let $Z_{ij}^{(1)}=\Tr\sigma_{\eta_{0}}^{(1)}B_{j}^{(1)}B_{i}^{(1)}$
and $Z_{ij}^{(2)}=\Tr\sigma_{\xi_{0}}^{(2)}B_{j}^{(2)}B_{i}^{(2)}$.
Because $\left\{ B_{i}^{(1)}\otimes I\right\} _{i=1}^{2}$ and $\left\{ I\otimes B_{i}^{(2)}\right\} _{i=1}^{2}$
are independent, $Z(B^{(\beta^{*})})=Z^{(1)}+Z^{(2)}.$ Note that
$Z^{(1)}$ is a real matrix. Let $\ket{\psi^{(3)}},\ket{l_{1}^{(3)}},\ket{l_{2}^{(3)}}\in\H_{3}$
be vectors in a Hilbert space $\H_{3}=\C^{2}$ such that $\braket{\psi^{(3)}}{l_{1}^{(3)}}=\braket{\psi^{(3)}}{l_{2}^{(3)}}=0$,
\begin{equation}
\braket{l_{j}^{(3)}}{l_{i}^{(3)}}=\left(V^{(3)}-\ii{\rm Im}Z^{(2)}\right)_{ij},
\end{equation}
 and $\left\Vert \ket{\psi^{(3)}}\right\Vert =1$ with a positive
real matrix $V^{(3)}=\sqrt{G}^{-1}\left|\sqrt{G}{\rm Im}Z^{(2)}\sqrt{G}\right|\sqrt{G}^{-1}.$
Because 
\begin{align}
\ket{\hat{\psi}^{(3)}} & :=\ket{\psi_{\xi_{0}}}\otimes\ket{\psi^{(3)}}\\
\ket{\hat{l}_{i}^{(3)}} & :=B_{i}^{(2)}\ket{\psi_{\xi_{0}}}\otimes\ket{\psi^{(3)}}+\ket{\psi_{\xi_{0}}}\otimes\ket{l_{i}^{(3)}}\qquad(i=1,2)
\end{align}
satisfy $\braket{\hat{\psi}^{(3)}}{\hat{l}_{1}^{(3)}}=\braket{\hat{\psi}^{(3)}}{\hat{l}_{2}^{(3)}}=0$
and $\braket{\hat{l}_{j}^{(3)}}{\hat{l}_{i}^{(3)}}=\left({\rm Re}Z^{(2)}+V^{(3)}\right)_{ij}\in\R$,
there exist an orthonormal basis $\{\ket k\}_{k=1}^{\dim\H_{2}\otimes\H_{3}}$
of $\H_{2}\otimes\H_{3}$ such that $\braket k{\hat{\psi}^{(3)}},\braket k{\hat{l}_{1}^{(3)}},\braket k{\hat{l}_{2}^{(3)}}$
are real numbers and $\braket k{\hat{\psi}^{(3)}}\not=0$ for $1\leq k\leq\dim\H_{2}\otimes\H_{3}$.
It can be seen that two observables 
\begin{equation}
\hat{B}_{i}=\theta_{0}^{i}+B_{i}^{(1)}\otimes I+I\otimes\left[\sum_{k}\frac{\braket k{\hat{l}_{i}^{(3)}}}{\braket k{\hat{\psi}^{(3)}}}\ket k\bra k\right]
\end{equation}
($i=1,2$) can be measured simultaneously, and they satisfy locally
unbiased conditions and achieve the Holevo bound, i.e.,
\begin{align}
\Tr(\rho_{\theta_{0}}\otimes\sigma^{(3)})\hat{B}_{i} & =\theta_{0}^{i}\qquad(1\leq i\leq2)\\
\Tr(\partial_{i}\rho_{\theta_{0}}\otimes\sigma^{(3)})\hat{B}_{j} & =\delta_{ij}\qquad(1\leq i,j\leq2)\\
\Tr(\rho_{\theta_{0}}\otimes\sigma^{(3)})\left(\hat{B}_{i}-\theta_{0}^{i}\right)\left(\hat{B}_{j}-\theta_{0}^{j}\right) & =({\rm Re}Z(B)+V^{(3)})_{ij}\qquad(1\leq i,j\leq2),
\end{align}
where $\sigma^{(3)}=\ket{\psi^{(3)}}\bra{\psi^{(3)}}$.}
\end{example}

\section{\label{sec:Conclusion}Conclusion}

In this paper, we focused on a logarithmic derivative $L_{i}^{(\beta)}$
lies between SLD $L_{i}^{(S)}$ and RLD $L_{i}^{(R)}$ with $\beta\in[0,1]$
to obtain lower bounds of weighted trace of covariance $\Tr GV_{\theta_{0}}[M,\hat{\theta}]$
of a locally unbiased estimator $(M,\hat{\theta})$ at $\theta_{0}$
of a parametric family of quantum states. We showed that all monotone
metrics induce lower bounds of $\Tr GV_{\theta_{0}}[M,\hat{\theta}]$,
and the maximum logarithmic derivative bound $\max_{0\leq\beta\leq1}C_{\theta_{0},G}^{(\beta)}$
is the largest bound among them. We showed that $\max_{0\leq\beta\leq1}C_{\theta_{0},G}^{(\beta)}$
has explicit solution when the $d$ dimensional model has $d+1$ dimensional
real space $\tilde{\T}\supset{\rm span}_{\R}\{L_{i}^{(S)}\}_{i=1}^{d}$
such that $\D_{\rho_{\theta_{0}}}(\tilde{\T})\subset\tilde{\T}$ at
$\theta_{0}\in\Theta$. Furthermore, when $d=2$, we showed that the
maximization problem $\max_{0\leq\beta\leq1}C_{\theta_{0},G}^{(\beta)}$
is the Lagrangian duality of the minimization problem to define Holevo
bound, and is the same as the Holevo bound. 
This explicit solution is the generalization of the solution \eqref{eq:suzuki_bound_ori} given for a two dimensional Hilbert space. 
%given in \cite{suzukiHolevo}. 

\section*{Acknowledgment}

%\begin{acknowledgments}
The author is grateful to Prof. A. Fujiwara for valuable comments.
%\end{acknowledgments}

\appendix

\section{Proof of (\ref{eq:trabs})\label{sec:trabs_proof}}

In this appendix, we give a proof of (\ref{eq:trabs}). 
\begin{lem}
For a $d\times d$ positive complex matrix $J$ and a real positive
matrix $G$,
\begin{equation}
\min\left\{ \Tr GV;\,V\text{ is a \ensuremath{d\times d} real matrix such that }V\geq J\right\} =\Tr GJ+\Tr\left|\sqrt{G}{\rm Im}J\sqrt{G}\right|.
\end{equation}
\end{lem}

\begin{proof}
Let $\left\{ \ket i\right\} _{i=1}^{d}$ be normalized eigenvectors
of $\sqrt{G}{\rm Im}J\sqrt{G}$. For a $d\times d$ real matrix $V$
such that $V\geq J$, because
\begin{equation}
\sqrt{G}V\sqrt{G}\geq\sqrt{G}{\rm Re}J\sqrt{G}\pm\sqrt{-1}\sqrt{G}{\rm Im}J\sqrt{G},
\end{equation}
we have
\begin{equation}
\bra i\sqrt{G}V\sqrt{G}\ket i\geq\bra i\sqrt{G}{\rm Re}J\sqrt{G}\ket i+\left|\bra i\sqrt{G}{\rm Im}J\sqrt{G}\ket i\right|.
\end{equation}
Therefore we obtain the inequality
\begin{align}
\sum_{i=1}^{d}\bra i\sqrt{G}V\sqrt{G}\ket i & =\Tr\sqrt{G}V\sqrt{G}\\
 & \geq\sum_{i=1}^{d}\{\bra i\sqrt{G}{\rm Re}J\sqrt{G}\ket i+\left|\bra i\sqrt{G}{\rm Im}J\sqrt{G}\ket i\right|\}\\
 & =\Tr\sqrt{G}{\rm Re}J\sqrt{G}+\Tr\left|\sqrt{G}{\rm Im}J\sqrt{G}\right|.
\end{align}
The equality is achieved when
\begin{equation}
V={\rm Re}J+\sqrt{G^{-1}}\left|\sqrt{G}{\rm Im}J\sqrt{G}\right|\sqrt{G^{-1}}.
\end{equation}
\end{proof}

\section{Derivation of Holevo bound \label{sec:Holevo_proof}}

In this appendix, it is proved briefly that the Holevo bound is lower
than the weighted trace of the covariance of any unbiased estimator.
\begin{thm}
Let $\S=\left\{ \rho_{\theta};\,\theta\in\Theta\subset\R^{d}\right\} $
be a smooth parametric family of density operators on a finite dimensional
Hilbert space $\H$. For a locally unbiased estimator $(M,\hat{\theta})$
at $\theta_{0}$ and a $d\times d$ positive real matrix $G$, 
\begin{equation}
\Tr GV_{\theta_{0}}[M,\hat{\theta}]\geq C_{\theta_{0},G}^{(H)},\label{eq:holevo_appendix}
\end{equation}
where $C_{\theta_{0},G}^{(H)}$ is the Holevo bound defined by (\ref{eq:holevo_bound}).
\end{thm}

\begin{proof}
Let
\begin{equation}
X^{M,i}:=\sum_{x\in\X}(\hat{\theta}^{i}(x)-\theta_{0}^{i})M_{x}.
\end{equation}
It follows that
\begin{equation}
V_{\theta_{0}}[M,\hat{\theta}]_{ij}-Z(X^{M})_{ij}=\sum_{x\in\X}\Tr\left[(\hat{\theta}^{i}(x)-\theta_{0}^{i}-X^{M,i})\right]\rho_{\theta_{0}}\left[(\hat{\theta}^{j}(x)-\theta_{0}^{j}-X^{M,j})\right]M_{x},
\end{equation}
thus
\begin{equation}
V_{\theta_{0}}[M,\hat{\theta}]\geq Z(X^{M}).
\end{equation}
Further, $X^{M,i}$ satisfies $\Tr\partial_{i}\rho_{\theta_{0}}X^{M,j}=\delta_{i}^{j}$.
Then (\ref{eq:holevo_appendix}) is proved.
\end{proof}

\section{\label{sec:schur}Schur complement}

In this paper, we utilize Schur complement. Suppose $A_{1},A_{2},A_{3}$
are $p\times p,p\times q,q\times q$ complex matrices such that $A_{3}$ is invertible. 
The Schur complement of the block $A_{3}$ of $A=\begin{pmatrix}A_{1} & A_{2}^{*}\\
A_{2} & A_{3}
\end{pmatrix}$ is defined by
\begin{equation}
A/A_{3}:=A_{1}-A_{2}^{*}A_{3}^{-1}A_{2}.
\end{equation}
The matrix $A$ can be decomposed as
\begin{align*}
A=\begin{pmatrix}A_{1} & A_{2}^{*}\\
A_{2} & A_{3}
\end{pmatrix} & =\begin{pmatrix}I & A_{2}^{*}A_{3}^{-1}\\
0 & I
\end{pmatrix}\begin{pmatrix}A/A_{3} & 0\\
0 & A_{3}
\end{pmatrix}\begin{pmatrix}I & 0\\
A_{3}^{-1}A_{2} & I
\end{pmatrix}.
\end{align*}
Therefore ${\rm rank}A={\rm rank}A_{3}$ if and only if $A/A_{3}=0$.
When $A$ is invertible, $A/A_{3}$ is also invertible and
\begin{equation}
A^{-1}=\begin{pmatrix}\left(A/A_{3}\right)^{-1} & -\left(A/A_{3}\right)^{-1}A_{2}^{*}A_{3}^{-1}\\
-A_{3}^{-1}A_{2}\left(A/A_{3}\right)^{-1} & A_{3}^{-1}+A_{3}^{-1}A_{2}\left(A/A_{3}\right)^{-1}A_{2}^{*}A_{3}^{-1}
\end{pmatrix}.
\end{equation}

\section{Multiple inner products and $\protect\D$ invariant space \label{sec:Dinv}}

In quantum statistics, multiple inner products are used. 
The commutation operator $\D$ can link the multiple monotone metrics. 
%In general,
%another inner product on a Hilbert space $\K$ with 
%an inner product $\left\langle \cdot,\cdot\right\rangle $ corresponds to an strictly
%positive operator. 
{\color{black}
In general,
%there is a one-to-one correspondence between positive operators and inner products on a Hilbert space $\K$ with a fixed inner product $\left\langle \cdot,\cdot\right\rangle$, i.e.,
any inner product $(\cdot,\cdot)$
on a Hilbert space $\K$ with a fixed inner product $\left\langle \cdot,\cdot\right\rangle$
has a positive operator $S$ uniquely
such that 
$(v,w)=\left\langle v,S w \right\rangle$
for $v,w\in \K$.
}
%any inner product $(\cdot,\cdot)$ on a Hilbert space $\K$ with an inner product $\left\langle \cdot,\cdot\right\rangle$ 
%corresponds to an strictly
%positive operator. 
For a strictly positive operator $S$ on the Hilbert
space $\K$, the following Lemma holds. 
\begin{lem}
\label{lem:genaral_pos}Let $\K_{0}\subset\K$ be a linear subspace
of $\K$, and let $\{e_{i}\}_{i=1}^{d}$ be a basis of $\K_{0}$.
The following conditions are equivalent:
\end{lem}

\begin{description}
\item [{(i)}] $S(\K_{0})=\K_{0}.$
\item [{(ii)}] $K^{(3)^{-1}}=K^{(2)^{-1}}K^{(1)}K^{(2)^{-1}}$, where $K^{(1)},K^{(2)},K^{(3)}$
are $d\times d$ matrix defined by $K_{ij}^{(1)}=\left\langle e_{i},Se_{j}\right\rangle $,
$K_{ij}^{(2)}=\left\langle e_{i},e_{j}\right\rangle $, $K_{ij}^{(3)}=\left\langle e_{i},S^{-1}e_{j}\right\rangle $. 
\end{description}
\begin{proof}
The Gram matrix of $\{Se_{i}\}_{i=1}^{d}\cup\{S^{-1}e_{i}\}_{i=1}^{d}$
with respect to the inner product $\left\langle \cdot,\cdot\right\rangle $
is
\begin{equation}
K=\begin{pmatrix}K^{(1)} & K^{(2)}\\
K^{(2)} & K^{(3)}
\end{pmatrix}.
\end{equation}
The condition (i) is equivalent to ${\rm rank}K=d$ because $\dim\left\{ {\rm span}\{Se_{i}\}_{i=1}^{d}\cup\{S^{-1}e_{i}\}_{i=1}^{d}\right\} =d$,
and is equivalent to
\begin{equation}
K/K^{(3)}=K^{(1)}-K^{(2)}K^{(3)^{-1}}K^{(2)}=0,
\end{equation}
where $K/K^{(3)}$ is the Schur complement given in Appendix \ref{sec:schur}. 
\end{proof}
By using Lemma \ref{lem:genaral_pos}, we can obtain a useful property
of $\D$ invariant space. Let $\S=\left\{ \rho_{\theta};\,\theta\in\Theta\subset\R^{d}\right\} $
be a smooth parametric family of density operators on a finite dimensional
Hilbert space $\H$. Let $\D_{\rho_{\theta_{0}}}:\B(\H)\to\B(\H)$
be the commutation operator with respect to a faithful state $\rho_{\theta_{0}}\in\S$.
Through a positive super operator $I+\ii\beta\D_{\rho_{\theta_{0}}}$,
the $\beta$ logarithmic derivatives $\{L_{i}^{(\beta)}\}_{i=1}^{d}$
and the corresponding inner product $\left\langle \cdot,\cdot\right\rangle _{\rho_{\theta_{0}}}^{(\beta)}$
are linked by (\ref{eq:bLD_SLD}) and (\ref{eq:bLD_SLD_inner}). From
these relations, we have the following lemma.
\begin{lem}
\label{lem:Dinv}The following conditions are equivalent:

\noindent\begin{minipage}[t]{1\columnwidth}%
\begin{description}
\item [{(i)}] ${\rm span_{\R}}\left\{ \partial_{i}\rho_{\theta_{0}}\right\} _{i=1}^{d}$
is $\D_{\rho_{\theta_{0}}}$ invariant. 
\item [{(ii)}] ${\rm span_{\R}}\left\{ L_{i}^{(S)}\right\} _{i=1}^{d}$
is $\D_{\rho_{\theta_{0}}}$ invariant. 
\item [{(ii)'}] ${\rm span_{\C}}\left\{ L_{i}^{(S)}\right\} _{i=1}^{d}$
is $\D_{\rho_{\theta_{0}}}$ invariant. 
\item [{(iii)}] ${\rm span_{\C}}\left\{ L_{i}^{(S)}\right\} _{i=1}^{d}={\rm span_{\C}}\left\{ L_{i}^{(\beta)}\right\} _{i=1}^{d}$
for any $\beta\in[0,1].$
\item [{(iv)}] $J_{\theta_{0}}^{(\beta)^{-1}}=J_{\theta_{0}}^{(S)^{-1}}\left({\rm Re}Z+\beta\ii{\rm Im}Z\right)J_{\theta_{0}}^{(S)^{-1}}$
for any $\beta\in[0,1]$, where $Z=\left[\Tr L_{i}^{(S)}\rho_{\theta_{0}}L_{j}^{(S)}\right]_{ij}$
is a $d\times d$ matrix. 
\end{description}
\end{minipage}
\end{lem}

\begin{proof}
At first, let us prove (i)$\Leftrightarrow$(ii). Because super operators
$\D_{\rho_{\theta_{0}}}$ and $\left(\frac{\bL_{\rho_{\theta_{0}}}+\bR_{\rho_{\theta_{0}}}}{2}\right)^{-1}$
are commutative,
\begin{equation}
\left(\frac{\bL_{\rho_{\theta_{0}}}+\bR_{\rho_{\theta_{0}}}}{2}\right)^{-1}\circ\D_{\rho_{\theta_{0}}}\left(\partial_{i}\rho_{\theta_{0}}\right)=\D_{\rho_{\theta_{0}}}\circ\left(\frac{\bL_{\rho_{\theta_{0}}}+\bR_{\rho_{\theta_{0}}}}{2}\right)^{-1}\left(\partial_{i}\rho_{\theta_{0}}\right)=\D_{\rho_{\theta_{0}}}\left(L_{i}^{(S)}\right).
\end{equation}
Therefore $\D_{\rho_{\theta_{0}}}\left(\partial_{k}\rho_{\theta_{0}}\right)\in{\rm span_{\R}}\left\{ \partial_{i}\rho_{\theta_{0}}\right\} _{i=1}^{d}$
if and only if $\D_{\rho_{\theta_{0}}}\left(L_{k}^{(S)}\right)\in{\rm span_{\R}}\left\{ L_{i}^{(S)}\right\} _{i=1}^{d}$.

The proof of (ii)$\Leftrightarrow$(ii)' is trivial because $\D_{\rho_{\theta_{0}}}(X)$
is self-adjoint for any self-adjoint operator $X$.

The proof of (ii)'$\Leftrightarrow$(iii) is also trivial because
of (\ref{eq:bLD_SLD}). 

The proof of (iii)$\Leftrightarrow$(iv) is given 
by Lemma \ref{lem:genaral_pos}
{\color{black}
with a positive operator $I+\ii\beta\D_{\rho_{\theta_{0}}}$
}
because
\begin{align}
{\rm Re}Z_{ij}+\beta\ii{\rm Im}Z_{ij} & =\left\langle L_{i}^{(S)},(I+\ii\beta\D_{\rho_{\theta_{0}}})L_{j}^{(S)}\right\rangle ^{(0)},\\
J_{\theta_{0},ij}^{(S)} & =\left\langle L_{i}^{(S)},L_{j}^{(S)}\right\rangle ^{(0)},\\
J_{\theta_{0},ij}^{(\beta)} & =\left\langle L_{i}^{(S)},
(I+\ii\beta\D_{\rho_{\theta_{0}}})^{-1}L_{j}^{(S)}\right\rangle ^{(0)}.
\end{align}
\end{proof}

{\color{black}
\section{Relation between the bounds (\ref{eq:suzuki_bound_ori}) and (\ref{eq:mld_expli})
\label{sec:suzuki_mld}}

In this appendix, we show that the explicit form (\ref{eq:suzuki_bound_ori})
given for $\dim\H=2$ and $d=2$ can be derived from (\ref{eq:mld_expli}). 

The inequality condition in (\ref{eq:mld_expli})
\begin{equation}
\hat{\beta}=\frac{\Tr\left|\sqrt{G}{\rm Im}J_{\theta_{0}}^{(R)^{-1}}\sqrt{G}\right|}{2\Tr G\left\{ J_{\theta_{0}}^{(S)^{-1}}-{\rm Re}(J_{\theta_{0}}^{(R)^{-1}})\right\} }\geq1
\end{equation}
can be transformed into an inequality
\begin{equation}
2\Tr G{\rm Re}(J_{\theta_{0}}^{(R)^{-1}})+2\Tr\left|\sqrt{G}{\rm Im}J_{\theta_{0}}^{(R)^{-1}}\sqrt{G}\right|\geq2\Tr GJ_{\theta_{0}}^{(S)^{-1}}+\left|\sqrt{G}{\rm Im}J_{\theta_{0}}^{(R)^{-1}}\sqrt{G}\right|.\label{eq:suzuki_mld1}
\end{equation}
The left hand side of (\ref{eq:suzuki_mld1}) is equal to $2C_{\theta_{0},G}^{(R)}$.
The right hand side of (\ref{eq:suzuki_mld1}) is equal to $C_{\theta_{0},G}^{(S)}+C_{\theta_{0},G}^{(Z)}$
because
\begin{equation}
\Tr GJ_{\theta_{0}}^{(S)^{-1}}+\left|\sqrt{G}{\rm Im}J_{\theta_{0}}^{(R)^{-1}}\sqrt{G}\right|=\Tr GZ(L^{(S)})+\Tr\left|\sqrt{G}{\rm Im}Z(L^{(S)})\sqrt{G}\right|=C_{\theta_{0},G}^{(Z)},
\end{equation}
where
\begin{equation}
Z(L^{(S)})=A=J_{\theta_{0}}^{(S)^{-1}}+\sqrt{-1}{\rm Im}(J_{\theta_{0}}^{(R)^{-1}})
\end{equation}
given in (\ref{eq:A_relation}) is used. Therefore the inequality
(\ref{eq:suzuki_mld1}) is equivalent to the inequality condition
\begin{equation}
C_{\theta_{0},G}^{(R)}\geq\frac{C_{\theta_{0},G}^{(Z)}+C_{\theta_{0},G}^{(S)}}{2}
\end{equation}
in (\ref{eq:suzuki_bound_ori}).

Further, by using (\ref{eq:bLD_Ab}), (\ref{eq:A_relation}), and
(\ref{eq:b_relation}), we have

\begin{align}
C_{\theta_{0},G}^{(\hat{\beta})} & =\Tr G{\rm Re}A+\hat{\beta}\Tr\left|\sqrt{G}{\rm Im}A\sqrt{G}\right|-\hat{\beta}^{2}\bra bG\ket b\\
 & =\Tr G{\rm Re}A+\frac{\Tr\left|\sqrt{G}{\rm Im}A\sqrt{G}\right|}{2\bra bG\ket b}\Tr\left|\sqrt{G}{\rm Im}A\sqrt{G}\right|-\left(\frac{\Tr\left|\sqrt{G}{\rm Im}A\sqrt{G}\right|}{2\bra bG\ket b}\right)^{2}\bra bG\ket b\\
 & =\Tr G{\rm Re}A+\frac{\left(\Tr\left|\sqrt{G}{\rm Im}A\sqrt{G}\right|\right)^{2}}{4\bra bG\ket b}\\
 & =\Tr G{\rm Re}A-\bra bG\ket b+\Tr\left|\sqrt{G}{\rm Im}A\sqrt{G}\right|+\frac{\left(\Tr\left|\sqrt{G}{\rm Im}A\sqrt{G}\right|\right)^{2}}{4\bra bG\ket b}-\Tr\left|\sqrt{G}{\rm Im}A\sqrt{G}\right|+\bra bG\ket b\\
 & =C_{\theta_{0},G}^{(R)}+\frac{\left(\Tr\left|\sqrt{G}{\rm Im}A\sqrt{G}\right|-2\bra bG\ket b\right)^{2}}{4\bra bG\ket b}\\
 & =C_{\theta_{0},G}^{(R)}+\frac{\left(\frac{1}{2}\left(2\Tr G{\rm Re}A+\Tr\left|\sqrt{G}{\rm Im}A\sqrt{G}\right|\right)-\left(\Tr G{\rm Re}A+\bra bG\ket b\right)\right)^{2}}{\bra bG\ket b}\\
 & =C_{\theta_{0},G}^{(R)}+\frac{\left[\frac{1}{2}(C_{\theta_{0},G}^{(Z)}+C_{\theta_{0},G}^{(S)})-C_{\theta_{0},G}^{(R)}\right]^{2}}{C_{\theta_{0},G}^{(Z)}-C_{\theta_{0},G}^{(R)}}=C_{\theta_{0},G}^{(R)}+S_{\theta_{0},G}.
\end{align}
Thus, it is confirmed that (\ref{eq:suzuki_bound_ori}) can be derived
from (\ref{eq:mld_expli}).
}

\end{document}